\documentclass[reprint,superscriptaddress,longbibliography,nofootinbib]{revtex4-2}

\usepackage[T1]{fontenc}
\usepackage[utf8]{inputenc}

\usepackage{float}
\usepackage{graphicx}
\usepackage{epsfig,epstopdf}
\usepackage[table]{xcolor}

\usepackage[english]{babel}

\PassOptionsToPackage{hyphens}{url}
\usepackage[hyphenbreaks]{breakurl}

\usepackage{amsmath}
\usepackage{amsfonts}
\usepackage{amssymb,amsthm}
\usepackage{bbm}
\usepackage{mathtools}

\usepackage{physics}
\usepackage{xprintlen}

\usepackage{complexity}

\usepackage{enumerate}

\PassOptionsToPackage{pdftex,hyperfootnotes=false,pdfpagelabels}{hyperref}
\usepackage{hyperref}
\hypersetup{
    colorlinks=true, linktocpage=true, pdfstartpage=3, pdfstartview=FitV,
    breaklinks=true, pdfpagemode=UseNone, pageanchor=true, pdfpagemode=UseOutlines,
    plainpages=false, bookmarksnumbered, bookmarksopen=true, bookmarksopenlevel=1,
    hypertexnames=true, pdfhighlight=/O,
    urlcolor=blue!60!black, linkcolor=blue!60!black, citecolor=gray, 
}   

\usepackage[capitalise,compress]{cleveref}

\makeatletter 
\hypersetup{
      pdftitle = {
            Sharp complexity phase transitions generated by entanglement
            },
      pdfauthor = {
            Abhinav Deshpande, 
            Bill Fefferman, 
            Soumik Ghosh, 
            Alexey V. Gorshkov, 
            Dominik Hangleiter
            },
        }
\makeatother

\definecolor{dominik}{rgb}{0.4,.0,0.6}

\newtheorem{theorem}{Theorem}
\newtheorem{proposition}[theorem]{Proposition}
\newtheorem{definition}[theorem]{Definition}
\newtheorem{lemma}[theorem]{Lemma}

\newtheorem{corollary}[theorem]{Corollary}
\newtheorem*{remark}{Remark}

\newcommand{\Y}{\mathsf{Y}}
\newcommand{\Z}{\mathsf{Z}}

\newclass{\np}{NP}
\newclass{\sP}{\#P}
\newclass{\postBQP}{postBQP}
\newclass{\sampBPP}{sampBPP}

\newcommand{\X}{\mathsf{X}}

\newcommand{\tw}{\mathsf{tw}}
\newcommand{\rw}{\mathsf{rw}}
\newcommand{\cw}{\mathsf{cw}}
\newcommand{\ew}{\mathsf{ew}}
\newcommand{\srw}{\mathsf{srw}}

\newcommand{\B}{\mathsf{B}}

\usepackage[all]{hypcap}
\usepackage{url}
\usepackage{textcomp, complexity}
\usepackage{tikz}
\usepackage{physics}
\usepackage{tabularx,environ}
\usepackage{multirow}
\usepackage{mathpazo}
\usepackage{mathrsfs}

\let\oldnl\nl
\newcommand{\nonl}{\renewcommand{\nl}{\let\nl\oldnl}}
\newcommand{\ccol}[2]{\color{#1}#2\color{black}}

\usepackage[capitalise,compress]{cleveref}
\crefname{section}{Section}{Section} 
\crefname{subsection}{Section}{Section}

\newtheoremstyle{mystyle}
{3pt}
{3pt}
{\upshape}
{}
{\bfseries}
{.}
{.5em}
{}

\theoremstyle{theorem}
\crefname{thm}{Theorem}{Theorems}
\Crefname{thm}{Theorem}{Theorems}

\crefname{corll}{Corollary}{Corollaries}
\theoremstyle{remark}

\theoremstyle{definition}

\theoremstyle{mystyle}
\newtheorem{todoaux}{To-Do}

\NewDocumentEnvironment{todo}{o}
 {\IfNoValueTF{#1}
   {\todoaux\addcontentsline{toc}{subsection}{\protect\numberline{\thesubsection}\ccol{red}{To-Do item}}}
   {\todoaux[#1]\addcontentsline{toc}{subsection}{\protect\numberline{\thesubsection}{\ccol{red}{To-Do item}} (#1)}}
   \ignorespaces}
 {\endtodoaux}

\makeatletter

\newcolumntype{\expand}{}
\long\@namedef{NC@rewrite@\string\expand}{\expandafter\NC@find}

\usepackage{environ}

\NewEnviron{task}[2][]{
  \def\problem@arg{#1}
  \def\problem@framed{framed}
  \def\problem@lined{lined}
  \def\problem@doublelined{doublelined}
  \ifx\problem@arg\@empty
    \def\problem@hline{}
  \else
    \ifx\problem@arg\problem@doublelined
      \def\problem@hline{\hline\hline}
    \else
      \def\problem@hline{\hline}
    \fi
  \fi
  \ifx\problem@arg\problem@framed
    \def\problem@tablelayout{|>{\bfseries}lX|c}
    \def\problem@title{\multicolumn{2}{|l|}{
        \raisebox{-\fboxsep}{\textsc{#2}}
      }}
  \else
    \def\problem@tablelayout{>{\bfseries}lXc}
    \def\problem@title{\multicolumn{2}{l}{
        \raisebox{-\fboxsep}{\textsc{#2}}
      }}
  \fi
\par
\noindent
  
  \begin{tabularx}{\columnwidth}{\expand\problem@tablelayout}
    \problem@hline
    \problem@title\\[\fboxsep]
    \BODY 
  \end{tabularx}
    \\
}

\begin{document}

\title{Sharp complexity phase transitions generated by entanglement}

\author{Soumik Ghosh}
\affiliation{Department of Computer Science, University of Chicago, Chicago, Illinois 60637, USA}

\author{Abhinav Deshpande}
\affiliation{Institute for Quantum Information and Matter, California Institute of Technology, Pasadena, California 91125, USA}

\author{Dominik Hangleiter}
\affiliation{Joint Center for Quantum Information and Computer Science and Joint Quantum Institute, University of Maryland \& NIST, College Park, Maryland 20742, USA}

\author{Alexey V.\ Gorshkov}
\affiliation{Joint Center for Quantum Information and Computer Science and Joint Quantum Institute, University of Maryland \& NIST, College Park, Maryland 20742, USA}

\author{Bill Fefferman}
\affiliation{Department of Computer Science, University of Chicago, Chicago, Illinois 60637, USA}

\date{\today}

\begin{abstract}
Entanglement is one of the physical properties of quantum systems responsible for the computational hardness of simulating quantum systems. But while the runtime of specific algorithms, notably tensor network algorithms, explicitly depends on the amount of entanglement in the system, it is unknown whether this connection runs deeper and entanglement can also cause inherent, algorithm-independent complexity.
In this work, we quantitatively connect the entanglement present in certain quantum systems to the computational complexity of simulating those systems.
Moreover, we completely characterize the entanglement and complexity as a function of a system parameter. Specifically, we consider the task of simulating single-qubit measurements of $k$--regular graph states on $n$ qubits. We show that, as the regularity parameter is increased from $1$ to $n-1$, there is a sharp transition from an easy regime with low entanglement to a hard regime with high entanglement at $k=3$, and a transition back to easy and low entanglement at $k=n-3$. As a key technical result, we prove a duality for the simulation complexity of regular graph states between low and high regularity.
\end{abstract}

\maketitle

\noindent A fundamental question since the inception of quantum computing has been to understand the physical mechanisms underlying the computational speedup of quantum computers. One of the most widely studied resources for a quantum speedup is entanglement \cite{Vidal2003,Jozsa2003}. 
However, understanding precisely \emph{how much} entanglement is necessary \emph{and} sufficient for a quantum system to be intractable to arbitrary classical simulation techniques has remained elusive. 
Quantum computations involving next to no entanglement can be hard to simulate classically \cite{Knill1998,Biham2004,Datta2005} and
relatively little entanglement can be universal for quantum computation \cite{Raussendorf2002,Hein2006,VandenNest2013}, while states with very high entanglement can be useless for quantum computation \cite{Bremner2009,Gross2009}.

One way the relation between entanglement and hardness has been studied is by considering the performance of specific simulation methods, like tensor networks \cite{Markov2008,Shi2006,Huang2020b,Pan2022}. The runtime of tensor-network algorithms depends exponentially on the amount of a certain type of entanglement \cite{Vidal2003,Shi2006,Markov2008}, as it determines how efficiently we can contract the tensor network. However, it is an open problem to characterize the situations in which tensor network algorithms are optimal. 
When can we find another algorithm that could do better in situations in which tensor networks are inefficient?
Moreover, when does the failure of tensor networks coincide with an inherent hardness of the problem itself? 
This essentially is the content of the second of Aaronson's ``Ten Semi-Grand Challenges for Quantum Computing Theory'' \cite{Aaronson2005b}.

The effect of the presence of entanglement on the hardness of classical simulation has been considered in various settings including  measurement-based quantum computing (MBQC) \cite{VandenNest2004,Nest2007a,Bremner2009,Gross2009}, the one-clean-qubit model \cite{Yoganathan2019}, and more recently in a line of research considering the time evolution under certain classes of Hamiltonians \cite{Bouland2016,Maskara2022}.
However, we are lacking a \emph{quantitative connection} between the entanglement present in certain quantum states and the inherent computational complexity of simulating those states.

\begin{figure*}
    \includegraphics{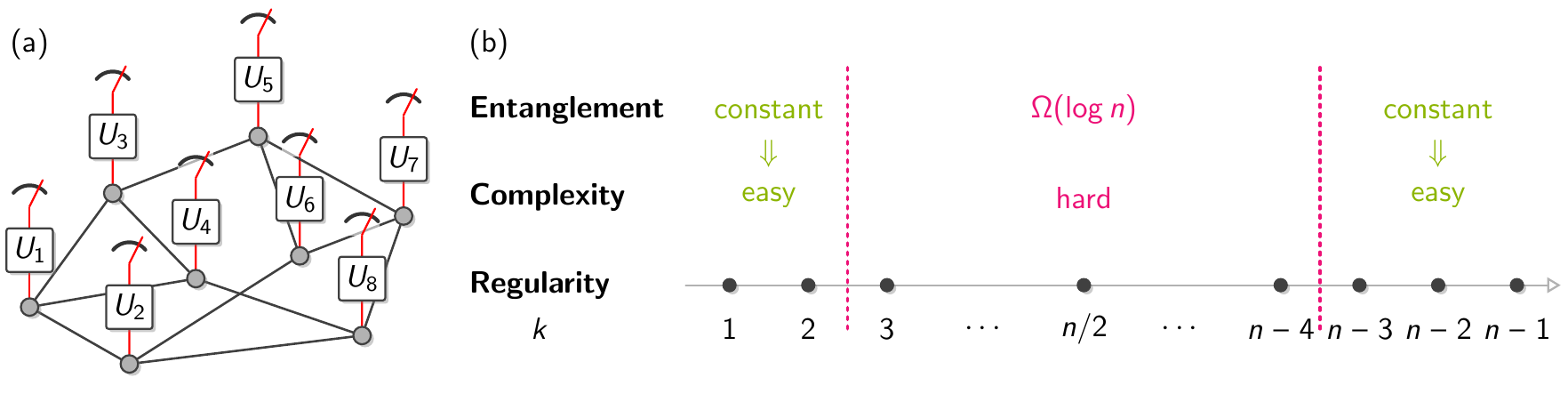}
    \caption{(a)~The family of quantum states we consider are graph states on a $k$--regular graph $G$ on $n$ qubits with arbitrary single qubit rotations $U_1, U_2, \ldots, U_n$. 
    The measurements are done in the standard basis. 
    (b)~Phase transitions of the entanglement (as measured by entanglement width) and computational complexity---whether classical simulation is easy or hard---as a function of the regularity parameter $k$. For both the entanglement width and the computational complexity, we take the worst case over all $k$--regular graphs $G$ as well as $U_1, U_2, \ldots, U_n$. 
    }
    \label{fig:setup}
\end{figure*}

In this paper, we answer Aaronson's question quantitatively with respect to the entanglement of regular graph states. For a simple graph $G$\,$=$\,$(V,E)$ given by the pair of vertex set $V$ and edge set $E$, the corresponding graph state $\ket{G}$ is defined as
\begin{equation}
    \ket{G} = \underset{(i, j) \in E}{\prod} (\mathsf{CZ})_{i, j} | + \rangle^{\otimes n},
\end{equation}
where $\mathsf{CZ}_{i, j}$ is the controlled-$\mathsf{Z}$ operator acting on vertices $i$ and $j$. The action of the $\mathsf{CZ}_{i, j}$ gate is invariant with respect to changing the control and the target qubits.
Graph states \cite{Hein2004} are a very well-motivated class to investigate the 
{interplay} of classical simulability {and entanglement}. 
On the one hand, a graph state directly maps to a tensor network, and one can invoke the measurement-based model of quantum computing \cite{Raussendorf2002,Raussendorf2006,Nielsen2006} to argue that certain graph states are not efficiently simulable and are, moreover, universal resources for quantum computations.
On the other hand, their entanglement can be conveniently analyzed using graph theory \cite{Hein2006}. 

Examples of universal resource states are graph states on hexagonal, square, or triangular lattices \cite{Raussendorf2003,Nest2006}. 
Under closed boundary conditions these resource states precisely correspond to $3$--, $4$--, and $6$--regular graphs, respectively. Conversely, graph states on a $2$--regular graph, i.e., a one-dimensional cluster state, and the graph state on an $(n$\,$-$\,$1)$--regular graph on $n$ qubits, i.e., the complete graph, are also well studied: both are efficiently simulable and at the same time have low entanglement \cite{Nest2006,Hein2006,Markham2007}.
However, for all other values of the regularity parameter $k$, it is unknown exactly when, if at all, classical simulation is intractable, and how the regularity parameter relates to the entanglement of the corresponding graph state.
\\

\noindent \textit{Our contributions.}---In this work, we completely characterize the computational complexity of simulating $k$--regular graph states in arbitrary product bases and their entanglement as a function of the regularity parameter $k$; see \cref{fig:setup}. 
We also identify new resource states for MBQC, which 
is a result of independent interest. 
Indeed, our constructions reach all the way to almost fully connected graphs that may be more natural for some experimental architectures such as ion traps~\cite{blatt_quantum_2012} or cavity quantum electrodynamics~\cite{Swingle2016} than low-degree lattices.
We also obtain new bounds in graph theory from complexity-theoretic assumptions.

Our setup includes arbitrary single-qubit gates at the end to perform the measurement in arbitrary local bases. This ensures that classical simulation algorithms that exploit specific properties---in particular, low stabilizer rank or $\mathsf{T}$-count \cite{Bravyi2016,Bravyi2019b} and low negativity in quasiprobability representations \cite{Mari2012,Pashayan2015,Delfosse2015, Raussendorf2017,Raussendorf2020}---are rendered inefficient.
Importantly, the last layer of local rotations does not affect the entanglement properties of the quantum state. 
In other words, the local rotations serve to isolate entanglement as the key causal factor responsible for hardness or easiness, respectively, and enable us to understand to what extent the entanglement present in a state serves as a necessary and sufficient criterion characterizing the simulation complexity. 

Our two main results can be summarized as follows and are illustrated in \cref{fig:setup}(b).

\begin{itemize}
\item As the regularity parameter $k$ is increased from its minimal value of $1$ to its maximal value of $n$\,$-$\,$1$, the simulation complexity first sharply changes from easy to provably hard precisely at $k$\,$=$\,$3$, but then changes sharply back to easy again at $k$\,$=$\,$n$\,$-$\,$3$. 

\item The entanglement scaling, as measured by the \emph{entanglement width} \cite{Nest2006}, is in one-to-one correspondence with the simulation complexity, changing from constant to at least logarithmic to constant at \emph{the same values of $k$} at which the simulation complexity changes from easy to hard and back to easy. 
\end{itemize}

Qualitatively, the entanglement width  measures the entanglement of  ``tree-like'' bipartitions of the state and this feature directly determines the runtime of tensor-network algorithms. 
It is also an LOCC (Local Operations and Classical Communication) monotone and hence a meaningful measure of entanglement~\cite{Raussendorf2002}.

We have thus identified a setup in which all the known easy cases are efficiently simulable using tensor-network algorithms precisely \emph{by virtue} of the state of the system having little entanglement. 
At the same time, all other cases are provably hard to simulate \emph{because} the entanglement present in the system facilitates universal measurement-based quantum computation, as we detail below.   
In this sense, the entanglement may justifiably be said to \emph{cause} the sharp complexity phase transitions. 
To the best of our knowledge, this is the first setup in which both features have been simultaneously demonstrated, and moreover, the entanglement and complexity transitions, as a function of a natural system parameter, are sharp. 

Finally, using appropriate complexity theoretic conjectures, bounds on entanglement width for the hard cases depicted in \cref{fig:setup}(b) can be improved from logarithmic to superlogarithmic or polynomial.
\\

\begin{figure*}
    \includegraphics{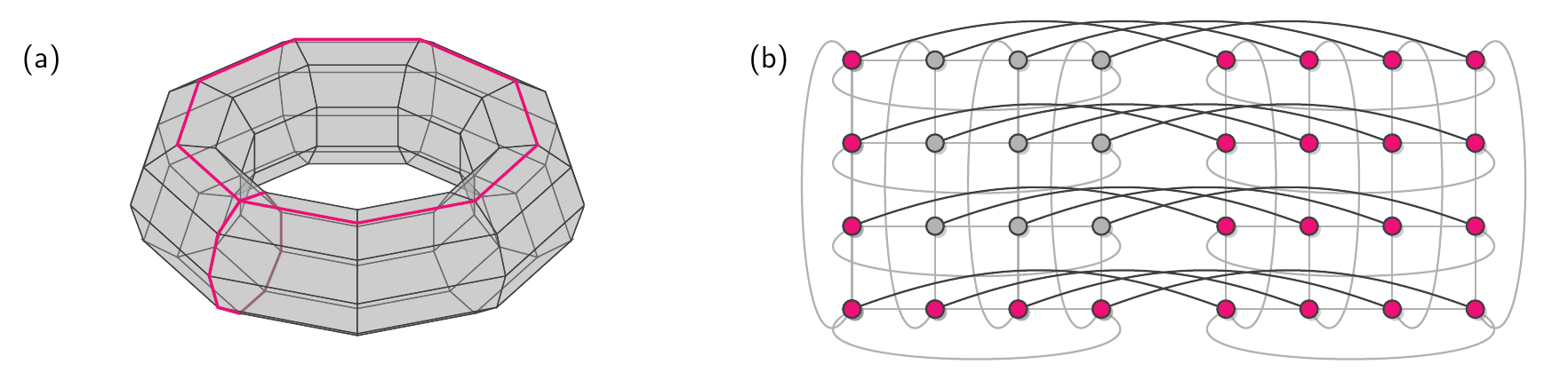}
    \caption{(a) A grid graph with closed boundary conditions is a torus, which is a $4$--regular graph. This is a resource state for MBQC because ``cutting open'' the torus along the pink lines gives back a grid graph.
    (b) Two tori connected together to construct a $5$--regular graph. The pink vertices are the ones we delete to recover a grid graph, which proves that this is a valid resource state for MBQC.
    \label{fig:hard graph}
    }
\end{figure*}

\noindent \textit{Main results.}---Let us now state our main results. 
We consider simulation of the quantum states in terms of both sampling from their output distributions \emph{and} computing their output probabilities up to constant multiplicative error in an arbitrary local product basis.
We also stress that all our results carry over to the case of approximate sampling, assuming appropriate complexity-theoretic conjectures (see the Supplemental Material~\cite{SM} for details).

Indeed, in the case in which simulation is hard, the two notions of simulation are intricately linked: 
given that computing output probabilities to constant multiplicative error is harder than any problem in the complexity class \sP, the sampling task cannot be efficiently solved. 
This can be shown by a standard reduction due to \textcite{Stockmeyer1983}.
In the hard regime, our proofs thus rely on showing \sP-hardness of estimating probabilities of a specific family of $k$--regular graphs in a specific family of local bases, implying the hardness of sampling.  
Conversely, easiness of sampling and computing output probabilities up to constant multiplicative error are independent properties and not implied by one another. 
However, our proofs in the easy regimes show that \emph{both} tasks are efficiently possible for our particular setup. 

Specifically, we prove the following results. 

\begin{theorem}[The easy regime]
\label{abstract thm:easiness}
In the regimes of very low ($k$\,$\leq$\,$2$) and very high ($k$\,$\geq$\,$n$\,$-$\,$3$) regularity, locally rotated $k$--regular graph states
(a) have constant entanglement width, and (b) can be simulated by a polynomial time classical algorithm. 
\end{theorem}
\noindent For all other values of $k$, we show that classical simulations are not efficiently possible: 
\begin{theorem}[The hard regime]
\label{abstract thm:hardness}
For every $3$\,$\leq$\,$k$\,$\leq$\,$n$\,$-$\,$4$, there exist locally rotated  $k$--regular graph states such that
(a) these states cannot be simulated classically in polynomial time, and
(b) the entanglement width scales at least logarithmically.
\end{theorem}
\noindent 
We also get the following corollary. 
\begin{corollary}
For every $3$\,$\leq$\,$k$\,$\leq$\,$n$\,$-$\,$4$, assuming $\BPP$\,$\subsetneq$\,$\P^{\#\P}$, there exist $k$--regular graph states satisfying Theorem \ref{abstract thm:hardness}(a) such that their entanglement width is superlogarithmic.

Assuming stronger hardness conjectures, the lower bounds on the entanglement width can be sharpened to $\Omega(n^\delta)$ for some constant $\delta$\,$>$\,$0$ (assuming the exponential time hypothesis) and to $\Omega(n^{1/2})$ (assuming the strong exponential time hypothesis.)
\end{corollary}

\begin{figure}[b]
    \centering
    \includegraphics{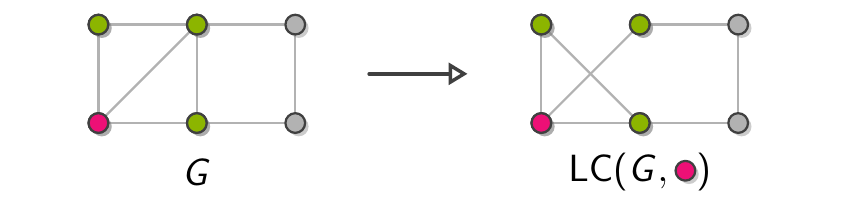}
    \caption{To perform local complementation $LC(G, a)$ of a graph $G$ with respect to vertex $a$ (pink), we take the complement of the subgraph comprising the neighbors of the pink vertex (green). }
    \label{neighbourhood flip}
\end{figure}

Let us note that our hardness results---while stated for the worst case---are in fact also valid on average over the local rotations via worst-to-average case reductions \cite{Bouland2018,Movassagh2019,Krovi2022}. 
Together, our results completely characterize the classical simulability of locally rotated regular graph states as a function of the regularity parameter in terms of both sampling and computing probabilities. \\

\begin{figure*}[]
    \includegraphics{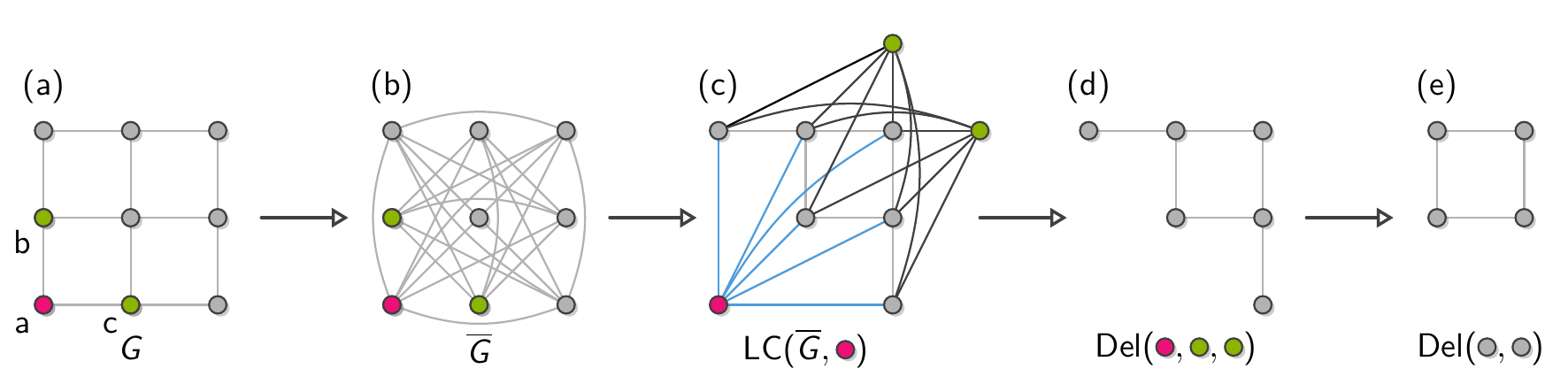}
    \caption{A visual proof that the complement of a grid graph is a resource state for MBQC. (a)~A $3$\,$\times$\,$3$ grid graph $G$. Consider (b)~$\overline{G}$---the complement of $G$. (c)~Apply a local complementation to vertex $a$. (d)~Delete vertices $b$ and $c$. (e)~Delete some of the gray vertices to finally reach a $2$\,$\times$\,$2$ grid graph.
    \label{fig:duality theorem}
    }
\end{figure*}

\noindent \textit{Proof of easiness results.}---In order to prove our easiness results, we utilize connections between entanglement width and classical simulations of graph states. Let us denote the entanglement width of a graph $G$ by $\ew(\ket{G})$; see Refs.~\cite{Nest2006,SM} for the precise definition. 

First, note that for $k$\,$\in$\,$\{1, 2, n$\,$-$\,$3, n$\,$-$\,$2, n$\,$-$\,$1 \}$, $\ew(\ket{G})$ is a constant for every $G$\,$\in$\,$\mathcal{G}_k$, where $\mathcal{G}_k$ is the set of all $k$--regular graphs. 
To see this, we make use of the relations of the entanglement width of a graph state $\ket G$ to width measures of the underlying graph $G$. 
In particular, the entanglement width is equal to the rank width of the underlying graph for graph states, {
and furthermore 
it can be related to the tree width and clique width of $G$  \cite{Nest2006}}. A refresher of these measures and their inter-relations are given in Section $F$ of the Supplemental Material~\cite{SM}.
All of these width measures express how ``tree-like'' the graph is from different perspectives. 
$1$-- and $2$--regular graphs have bounded tree width, which implies that they have bounded rank width and therefore also bounded entanglement width. Additionally, rank width, and hence, entanglement width, satisfy a duality property: if it is bounded for a graph $G$, it is also bounded for the complement $\overline G$ of $G$ \cite{Courcelle2000,Oum2017}. This fact allows us to argue that $(n$\,$-$\,$3)$-- and $(n$\,$-$\,$2)$--regular graphs have bounded entanglement width. 

Qualitatively, graph states with low entanglement width are efficiently simulable via tensor network simulation methods by the technique of Ref.~\cite{Nest2007a}. 
For a graph $G$, the idea is to construct a tree-tensor-network decomposition of a graph state $\ket{G}$. 
This takes time $\text{poly}(n, 2^{\ew(\ket{G})})$. 
Given this decomposition, and using techniques of Refs.~\cite{Markov2008,Shi2006,Nest2007a}, one can compute any output probability under any set of local rotations.
Additionally, one can also sample from the resulting output distributions.  \\

\noindent \textit{Hamming weight symmetry for the complete graph.} For the complete graph---i.e, the $(n$\,$-$\,$1)$--regular graph---we construct a new recursive algorithm that allow us to simulate arbitrary single-qubit product measurements. 

Specifically, our approach relies on an inherent symmetry  of the complete graph: the fact that any output probability of the complete graph on $n$ vertices has a Hamming weight symmetry---it can be written as a linear combination of {$n$\,$+$\,$1$ many terms}, one for each Hamming weight, such that each of them is efficiently computable. Using this fact, we design a recursion tree and show how to traverse it in polynomial time. The rigorous proof is given in Section $B$ of the Supplemental Material~\cite{SM} , along with a description of the recursion tree. 

While it is known that the output probabilities of the complete graph can be computed efficiently \cite{Nest2006,Markham2007,Hein2006}, to the best of our knowledge, our approach is novel and might have applications elsewhere to prove easiness, especially in problems having a Hamming weight symmetry. Some recent works have also used this symmetry to devise classical algorithms for quantum simulation \cite{Bravyi2020,Bringewatt2020}.
\\

\noindent \textit{Proof of hardness for $3$\,$\leq$\,$k$\,$\leq$\,$n/2$.} In order to prove our hardness results, we make use of the fact that certain graph states are resources for MBQC. 
Using Aaronson's result that $\postBQP=\PP$ \cite{Aaronson2005}, the output probabilities of a resource state for MBQC with local rotations are \sP-hard to compute \cite{Raussendorf2002,Schuch2007a,Fujii2017,Bermejo-Vega2018}. 
Then, using Stockmeyer's theorem, it is not possible to efficiently sample from their output distribution unless the polynomial hierarchy collapses \cite{Stockmeyer1983}; see \cite{Hangleiter2022} for an overview of this argument. 
In particular, this is true for the square lattice and the hexagonal lattice \cite{Raussendorf2003}.
{

Furthermore, we exploit the fact that certain single-qubit Clifford operations on a graph state $\ket G$, with classical communication and standard basis measurements, result in \emph{vertex deletion} and \emph{local complementation} of $G$ \cite{Dahlberg2018}. Local complementation flips the neighborhood of a vertex: connected vertices in the neighborhood are disconnected, and any two disconnected vertices are joined by an edge. This operation is illustrated in \cref{neighbourhood flip}.
It is known that if we can transform a parent graph $G$ to a hexagonal or grid graph by vertex deletion and local complementation, then $\ket G$ is a universal resource for MBQC and hence hard to simulate \cite{Nest2006,Fujii2017}. 
}

Our construction starts from the observation that hexagonal and square lattices with closed boundary conditions on the torus are, respectively, $3$-- and $4$--regular graphs. 
These are universal resources for MBQC, since we can reach planar hexagonal and square lattices by vertex deletion: we ``cut'' the torus open, see \cref{fig:hard graph}(a). Consequently, computing the output probabilities of $G$ in an arbitrary local basis is $\sP$-hard for $3$-- and $4$--regular graphs.

For graphs with higher regularity, we need more involved constructions. We reverse-engineer $k$--regular resources by starting from the $4$--regular resource state---the square lattice on a torus---and boost it up to $k$--regularity by adding gadgets, which can be removed by local complementation or vertex deletion. 
Note that, in most cases, the desired regularity scales with the total number of vertices (for example, consider an $n$--vertex, $n/2$--regular graph), and adding each new gadget might change the total number of vertices. So, adding too many gadgets might not always be a good construction strategy.

In light of this, starting from a grid graph on a torus, i.e., an $n$ vertex, $4$--regular graph, we add just a single gadget, namely another grid graph on a torus, see \cref{fig:hard graph}(b). We then judiciously connect the two grid graphs in a way such that every vertex is $k$--regular. It is nontrivial to argue that such a connection pattern even exists. We prove that it does using the Gale-Ryser theorem \cite{Gale1957,Ryser2009,Krause1996}, for every $4$\,$<$\,$k$\,$\leq$\,$n/2$. The Gale-Ryser theorem is constructive. Thus, our constructions prove that there exists an explicit $n$--vertex, $k$--regular graph $G$ such that computing the output probabilities of $G$ in an arbitrary local basis is \sP-hard, for every $4 \leq k \leq n/2$. \\

\noindent \textit{The duality property.}---Finally, we show that the complexity of simulating graphs with low regularity and graphs with high regularity satisfies a duality property.
Specifically, we prove that the complement of an $n$\,$\times$\,$n$ hexagonal graph or grid graph is a resource state for MBQC.
Hence, the corresponding $(n$\,$-$\,$4)$--regular graph state is universal under postselection, and simulating product measurements of it is classically intractable. 

To see this, consider an $n$\,$\times$\,$n$ grid graph $G$, and mark three vertices---a corner vertex of degree $2$, and its two neighbors, see \cref{fig:duality theorem}. 
Denote these vertices by $a$ (the pink vertex), $b$, and $c$ (the green vertices). Now, in the complement graph $\overline{G}$, apply \emph{local complementation} to vertex $a$, that is, we take the complement of the neighborhood of $a$. 
Then delete the vertices $a,b,c$, and subsequently, delete all the vertices in the same row and column as $a$ in $G$. 
We are left with an $(n$\,$-$\,$1)$\,$\times$\,$(n$\,$-$\,$1)$ grid graph, which is a resource state for MBQC. 
The intuition behind this is the following:
since the grid graph has bounded degree, for a vertex in the complement of such a graph, the only neighbors that are not connected are those which were connected in the original graph. 
Hence, local complementation mostly restores the original connections apart from those of $a$ and its neighbors. 
An analogous strategy shows that the complement of an $n$\,$\times$\,$n$ hexagonal lattice is also a resource state for MBQC. \\

\noindent \textit{Proof of hardness for $n/2$\,$<$\,$k$\,$\leq$\,$n$\,$-$\,$4$.}
We now extend our hardness proof to the regime of $n/2$\,$<$\,$k$\,$<$\,$n$\,$-$\,$4$. 
The idea is to take the hard graphs we constructed for $4$\,$\leq$\,$k$\,$\leq$\,$n/2$, comprising two copies of the grid graph on the torus, and then complement those hard graphs. 
If we started with a $k$--regular graph, after complementation, we are left with an $(n$\,$-$\,$k$\,$-$\,$1)$--regular graph. 
We then delete all vertices which were part of the second grid graph in the original graph and then apply local complementation to one of the vertices and vertex deletion in the column and row of that vertex. 

As a consequence, we obtain an explicit duality of simulation complexity between the regimes of high and low regularity. 
In other words, we find that there is an explicit $n$--vertex, $k$--regular graph $G$ such that computing the output probabilities of $G$ in an arbitrary local basis is \sP-hard, for every $n/2$\,$<$\,$k$\,$\leq$\,$n$\,$-$\,$4$.

Finally, we straightforwardly obtain bounds on the entanglement width of regular graphs in the easy regime using width measures from graph theory \cite{Corneil2001,Dabrowski2019,Fellows2009,Arnborg1987,Kaminski2009,Gurski2000,Golumbic1999,Chekuri2014, Diestel2017, Courcelle2000,Oum2017}, specifically tree width, rank width, and clique width, which can be related to the entanglement width. 
\\

\noindent \textit{Outlook.}---We have completely resolved Aaronson's question for regular graph states, going significantly beyond initial results on the interplay between simulability and entanglement in Refs.~\cite{VandenNest2004, Nest2006,Nest2007a}.
An immediate follow-up problem is to characterize the interplay between entanglement and simulation complexity of more restricted, physical families of graphs such as planar or bipartite graphs.
Our gadget constructions do not obviously generalize to these more restricted cases. 
As a result, we need new techniques to prove \sP-hardness. 

More generally, we can ask: 
can Aaronson's question of which systems are classically simulable be resolved in general, or even for slightly more general setups beyond graph states? 
Beyond graph states, the entanglement width is not always related to the classical simulation complexity of the corresponding quantum states. 
Furthermore, some of the qualitative interpretations of entanglement width, which were based on these measures, break down.
It thus remains a fascinating question to identify a metric that tracks the hardness in other settings. 
It is possible that situations, similar to the one we have identified, also exist for other measures such as the stabilizer rank of a state or its negativity, but it remains open if there is a universal single \emph{physical property} that fully determines the complexity of simulating a system.
More likely, simulation complexity is always a function of several different properties. 
\\

\begin{acknowledgments}
\noindent \textit{Author contributions.}---S.G. proved the results and wrote the initial draft of the manuscript. A.D., D.H., A.G. and B.F. contributed equally in helping to develop the setting, helping with the proofs, and finalizing the manuscript.
\\

\noindent \textit{Acknowledgments.}---We thank Joe Fitzsimmons for sharing his hints regarding the recursive algorithm for complete graphs, and Misha Lavrov for sharing his hint regarding constructing the gadget of the hardness proofs. S.~G. thanks Kunal Marwaha for helpful comments about the manuscript.
We are  grateful to the Simons Institute for the Theory of Computing, where parts of this work was conducted while some of the authors were visiting the institute.
A.~D.~acknowledges funding provided by the National Science Foundation RAISE-TAQS 1839204 and Amazon Web Services, AWS Quantum Program.
The Institute for Quantum Information and Matter is an NSF Physics Frontiers Center (NSF Grant PHY-1733907).
B.F. and S.G. acknowledge support from AFOSR (FA9550-21-1-0008). A.V.G.~was supported in part by the DoE ASCR Accelerated Research in Quantum Computing program (award No.~DE-SC0020312), NSF QLCI (award No.~OMA-2120757), DoE QSA, DoE ASCR Quantum Testbed Pathfinder program (award No.~DE-SC0019040), NSF PFCQC program, AFOSR, AFOSR MURI, ARO MURI, and DARPA SAVaNT ADVENT.   
This material is based upon work partially
supported by the National Science Foundation under Grant CCF-2044923 (CAREER) and by the U.S. Department of Energy, Office of Science, National Quantum Information Science Research Centers (Q-NEXT) as well as by DOE QuantISED grant DE-SC0020360.  This research was also supported in part by the National Science Foundation under Grant No. NSF PHY-1748958.
D.H.\ acknowledges financial support from the US Department of Defense through a QuICS Hartree Fellowship. 
\vspace{7pt}
\end{acknowledgments}

%

\cleardoublepage

\onecolumngrid
\cleardoublepage
\setcounter{page}{1}
\setcounter{equation}{0}
\setcounter{footnote}{0}
\setcounter{figure}{0}
\thispagestyle{empty}
\begin{center}
\textbf{\large Supplemental Material for ``Sharp complexity phase transitions generated by entanglement''}\\
\vspace{2ex}
Soumik Ghosh, Abhinav Deshpande, Dominik Hangleiter, Alexey V. Gorshkov, and Bill Fefferman
\vspace{2ex}
\end{center}

\twocolumngrid
\renewcommand\thesection{S\arabic{section}}
\renewcommand\thesection{S\arabic{section}}
\renewcommand\thefigure{S\arabic{figure}}

\tableofcontents
\section{Setting}

In this section, we define our setup and define some conventions that we use in this work. For a graph $G$, denote the corresponding graph state as $\ket{G}$. We consider graph states on $k$--regular graphs for $k \in [n-1] \equiv \{1,2, \ldots, n-1\}$. 

\begin{definition}[$k$--regular graph]
A $k$-\emph{regular graph} is a graph in which the degree of every vertex, that is, the number of adjacent vertices, is exactly $k$.
\end{definition} 

We then consider measurements of the qubits, or in other words, the vertices of $G$ in an arbitrary single-qubit basis. 
Such a measurement is equivalent to the application of a product of arbitrary single qubit gates followed by a standard basis measurement. 

For each qubit $i \in [n] \cong V$, we parameterize the single-qubit unitary $U_i$ as
\begin{equation}
    U_i(\theta_i, \phi_i) = \begin{pmatrix}
\cos\frac{\theta_i}{2} & -\sin \frac{\theta_i}{2} \\
e^{i \phi_i} \sin \frac{\theta_i}{2} & e^{i\phi_i} \cos \frac{\theta_i}{2}
\end{pmatrix}.
\end{equation}
When the arguments of $ U_i(\theta_i, \phi_i)$ are clear from the context, we drop them and just use $U_i$. Let
\begin{equation}
\label{local rotations}
    \mathsf{U} = U_1 \otimes U_2 \otimes \cdots \otimes U_n.
\end{equation}
For fixed $G$ and $\mathsf{U}$, the probability of getting any outcome $x \in \{0, 1\}^n$ is given by
\begin{equation}
    p_x (G, \mathsf{U}) = \bigg|\langle x| \underset{i=1}{\overset{n}{\otimes}} U_i \ket{G}\bigg|^2.
\end{equation}
A closely related quantity, the \emph{probability amplitude}, is defined as
\begin{equation}
    q_x (G, \mathsf{U}) = \langle x| \underset{i=1}{\overset{n}{\otimes}} U_i \ket{G}.
\end{equation}
When the context is clear from usage, we drop either $G$ and $\mathsf{U}$ or both from our notation. 
Additionally, let us denote by $\mathcal{D}({G, \mathsf{U}})$ the following probability distribution:
\begin{equation}
    \underset{X \sim \mathcal{D}({G, U^{(l)}})}{\mathsf{Pr}}[X = x] = p_x (G, \mathsf{U}).
\end{equation}
This probability distribution is over $x$, for a fixed $G$ and $\mathsf{U}$.

Equivalently, we can view our setting in the standard circuit picture. 
In this picture, the family of quantum circuits we are going to study contain a Clifford part, which is used to construct the graph state $\ket{G}$. 
In this part, there is a first layer of Hadamard gates, followed by a sequence of controlled-$\mathsf{Z}$ gates applied on the edges of $G$.
This is followed by the layer of single-qubit gates and a standard basis measurement, as illustrated in Figure $1$ of the main text.

Throughout this work, we let the symbols $\X, \Y$, and $\Z$ denote single qubit Pauli-$\X$, Pauli-$\Y$, and Pauli-$\Z$ gates respectively. Additionally, we use  capital letters, like $U, V$ et cetera, to denote quantum gates. 

\section{Techniques}
In this section, we talk about manipulating graph states and also talk about some useful complexity theoretic results that we use in this work. First, we start with basic concepts in graph theory and how they relate to graph states. Then, we discuss how we measure entanglement in the rest of the work. Finally, we introduce some advanced concepts in graph theory and use them to find upper and lower bounds on our entanglement measure.

\subsection{Some basic concepts in graph theory}

Let us begin by briefly rehashing some well-known concepts from graph theory. 

Consider a graph $G = (V, E)$, where $V$ is the set of vertices and $E$ is the set of edges. 
All the graphs, unless otherwise stated, are simple and undirected. 
Sometimes, to avoid ambiguity, we use $V_G$ and $E_G$ to denote that the vertex and edge sets of $G$. 
We state a series of standard definitions from graph theory below. We assume some familiarity with other related notions from graph theory, which can be found in \cite{Diestel2017}. 

\begin{definition}[Adjacency matrix]
The \emph{adjacency matrix} of a graph $G$ is a $|V| \times |V|$ matrix $A$ such that
\begin{equation}
\begin{aligned}
    A_{i, j} &= 1~~~~~~\text{if}~(u, v) \in E, \\
    &= 0~~~~~~\text{otherwise.}
\end{aligned}
\end{equation}
\end{definition}
\begin{definition}[Neighborhood]
The \emph{neighborhood} of a vertex $v$ in the graph $G$ is defined as
\begin{equation}
    N_v(G) = \{u : (u, v) \in E\}.
\end{equation}
\end{definition}
\noindent For two graphs $G$ and $H$ such that $H$ is a subgraph of $G$, the operation $K = G \setminus H$ is defined as
\begin{equation}
    K = (V_G, ~E_G \setminus E_H).
\end{equation}

\begin{definition}[Complete graph]
A \emph{complete graph} is a graph in which every pair of distinct vertices is connected by an edge. 
\end{definition}

\begin{definition}[Complement of a graph]
\label{complement}
The \emph{complement} of a graph $G$ with $n$ vertices, denoted by $\overline{G}$, is defined as
\begin{equation}
    \overline{G} = G_{\text{complete}} \setminus G, 
\end{equation}
where $G_{\text{complete}}$ is the complete graph on $n$ vertices.
\end{definition}
\begin{definition}[Subgraph]
A \emph{subgraph} of a graph $G = (V, E)$ is another graph whose vertex and edge sets are subsets of $V$ and $E$ respectively.
\end{definition}
\begin{definition}[Tree]
A \emph{tree} is a connected, acyclic graph.
\end{definition}
\begin{definition}[Binary tree]
A \emph{binary tree} is a tree where each vertex is connected to at most three other vertices.
\end{definition}

\begin{definition}[Bipartite graph]
A \emph{bipartite graph} is a graph whose vertices can be divided into two disjoint sets $U$ and $V$
such that every edge connects a vertex in $U$ to one in $V$.
\end{definition}

\begin{definition}[Complete bipartite graph]
A \emph{complete bipartite graph} is a bipartite graph whose vertices can be divided into two disjoint sets $U$ and $V$ such that 
\begin{enumerate}[(a)]
    \item every edge connects a vertex in $U$ to one in $V$, and 
    \item every vertex of $U$ is connected to every vertex of $V$.
\end{enumerate}
\end{definition}
A complete bipartite graph is denoted by $K_{p \times q}$, where $p$ and $q$ are the cardinalities of the sets $U$ and $V$, respectively.

\subsection{Operations on graph states}
\label{operations}
We frequently make use of a graph-theoretic operation called \emph{vertex deletion}. We define it as follows.

\begin{definition}[Vertex deletion]
For a graph $G = (V, E)$, the operation \emph{vertex deletion}, when applied to the vertex $v$ of $G$, produces a new graph $H$ such that
\begin{equation}
    H = (V \setminus v, E \setminus S),
\end{equation}
where $S$ is the set of all the edges incident to $v$.
\end{definition}
Applying vertex deletion to a vertex $v$ of $G$ is equivalent to measuring the corresponding qubit of $\ket{G}$ in the $\mathsf{Z}$ basis. Indeed, let $P_v^{(\mathsf{Z}, \pm)}$ be the two Pauli projectors onto the $\mathsf{Z}$ basis when we measure vertex $v$. Let $N(v)$ be the neighborhood of $v$ and let us delete the vertex $v$ to get the graph $H$. Then,
\begin{equation}
\label{gates}
\begin{aligned}
    P_v^{(\mathsf{Z}, +)} \ket{G} &= \frac{1}{\sqrt{2}} \ket{0} \ket{H}, \\
    P_v^{(\mathsf{Z}, -)} \ket{G} &= \frac{1}{\sqrt{2}} \ket{1} \prod_{u \in N_v(G)} Z_u \ket{H}.
\end{aligned}
\end{equation}

Note that the gates $Z_u$ in equation \eqref{gates} are single-qubit gates. Hence, by an appropriate choice of a last layer of local rotations, we can implement a vertex deletion. 

Before defining the next graph operation, let us define the symmetric difference operator between two sets.
\begin{definition}[Symmetric difference]
The symmetric difference operator between two sets $A$ and $B$, denoted by $\Delta$, is defined as
\begin{equation}
    A ~\Delta ~B = (B \setminus A) \cup (A \setminus B).
\end{equation}
\end{definition}
\begin{definition}[Local complementation]
For a graph $G = (V, E)$, a \emph{local complementation} $\tau_v$ on the vertex $v$ flips the neighborhood of $v$: two vertices which were previously connected (in the neighborhood of $v$) get disconnected and two vertices which were previously disconnected get connected. The neighborhood of a vertex $u$ in the graph $\tau_v(G)$ is given by
\begin{equation}
\begin{aligned}
    N_u(\tau_v(G)) &= N_u(G)~ \Delta ~(N_v \setminus {u})~~~~~~\text{if}~(u, v) \in E, \\
    &= N_u(G)~~~~~~\text{otherwise.}
\end{aligned}
\end{equation}
\end{definition}
It is known that applying a local complementation to a vertex $v$ in $G$ is equivalent to applying a sequence of single qubit Clifford unitaries
\begin{equation}
\label{local complementation}
    U = \exp\left(-i \frac{\pi}{4} X_v\right) \prod_{u \in N_v} \exp\left(i \frac{\pi}{4} Z_u\right)
\end{equation}
to the graph state $\ket{G}$. That is,
\begin{equation}
    \ket{\tau_v(G)} = U_v \ket{G}.
\end{equation}
Local complementation is illustrated in Figure $4$ of the main text.

\subsection{Some useful results from complexity theory}
\label{complexity theory results}
Before stating our technical lemmas, we state a few results from complexity theory for convenience. Variants of these statements and proofs have appeared in \cite{Aaronson2013a,FeffermanWilliamJason2014,Bouland2018,Hein2006,Raussendorf2002}, so we only state the lemmas without proofs. 

\begin{lemma}
\label{hardness grid graph}
    Consider an $n$--vertex graph $R_n$ such that the corresponding graph state $\ket{R_n}$ is a resource state for measurement based quantum computation. Consider local rotations $\mathsf{U}$, as defined in equation \eqref{local rotations}. Then, computing $p_x(R, \mathsf{U})$ is $\# \P$-hard in the worst case over $\mathsf{U}$, for any $x \in \{0, 1\}^n$, up to constant multiplicative error.
\end{lemma}

The proof follows from \cite{Raussendorf2002,Raussendorf2003,Fujii2017,Bermejo-Vega2018}. For example, considering graphs with $n$ vertices, this fact holds for an $\sqrt{n} \times \sqrt{n}$ grid graph. The idea is that resource states can do universal quantum computing under post-selection. By standard arguments sketched in these papers, that implies computing probabilities is $\# \P$-hard.

\begin{lemma}
\label{resource states}
Let $R$ be an $n$--vertex graph such that $\ket{R_n}$ is a resource state for measurement based quantum computation. Let $T$ be a graph such that $R$ can be reached from $T$ by vertex deletion and local complementation. Then, computing $p_x(T, \mathsf{U})$ is $\# \P$-hard in the worst case over $\mathsf{U}$, for every $x \in \{0, 1\}^n$, up to constant multiplicative error.
\end{lemma}

The proof follows from \cite{Fujii2017,Dahlberg2018,Dahlberg2020}. The intuition is that both vertex deletion and local complementation correspond to single-qubit operations, as we saw in \cref{operations}. Hence, $R$ is equivalent to a locally rotated grid graph, which is a known resource state.

\begin{lemma}[\cite{Stockmeyer1983}]
\label{lemma: sampling}
Let $C$ be a quantum circuit and let 
\begin{equation}
    p_x = |\langle x| C |0^n \rangle|^2,
\end{equation}
for $x \in \{0, 1\}^n$. Consider a distribution $\mathcal{D}$ given by
\begin{equation}
    \underset{X \sim \mathcal{D}}{\mathsf{Pr}}[X = x] = p_x.
\end{equation}
Note that this is a distribution over $x$ for a fixed $C$. Then, if there is a polynomial time classical algorithm to sample from $\mathcal{D}$, then there is a $\BPP^{\NP}$ algorithm to estimate $p_x$, upto constant multiplicative error, for a random choice of $x \in \{0, 1\}^n$. 

Consequently, if computing every $p_x$ is $\# \P$-hard upto constant multiplicative precision, no polynomial time classical algorithm exists to sample from $\mathcal{D}$, assuming the $\PH$ does not collapse to $\BPP^{\NP}$.
\end{lemma}

The proof follows from Stockmeyer's counting theorem \cite{Stockmeyer1983} and can also be found in \cite{Bouland2018}. \cref{lemma: sampling} states that assuming the $\PH$ does not collapse to $\BPP^{\NP}$, one can rule out the existence of classical exact samplers which sample from the output distribution of $C$. Through additional appropriate conjectures about the additive error hardness of computing the output probabilities of $C$, one can also extend the result of \cref{lemma: sampling} to rule out classical approximate samplers, up to an appropriate error in total variation distance, using standard techniques which are also outlined in \cite{Bouland2018,Bouland2022,Kondo2022}. Qualitatively, the proof says that the two popular characterizations of simulation of quantum circuits---sampling and estimation---are interlinked. The presence of a sampler in $\BPP$ implies the presence of an estimator in $\BPP^{\NP}$.

\subsection{Measuring entanglement entropy}
Let $|\psi\rangle$ be an $n$-qubit pure state and let $(\A, \B)$ be a partition of the qubits and let the corresponding Hilbert spaces be $\mathcal{H}_\A$ and $\mathcal{H}_\B$. Consider a Schmidt decomposition of $|\psi\rangle$ as follows:
\begin{equation}
    |\psi\rangle = \sum_{i=1}^{\mathsf{min}\left(2^{|\A|}, 2^{|\B|}\right)} \sqrt{s_i} \ket{u_i} \ket{v_i},
\end{equation}
where $\bigg\{ \ket{u_i} :i \in \{1, 2, \ldots, 2^{\mathcal{H}_|\A|}\} \bigg\}$ and $\bigg\{ \ket{v_i} :i \in \{1, 2, \ldots, 2^{\mathcal{H}_|\B|}\} \bigg\}$ are two sets of orthonormal basis vectors of the Hilbert spaces $\mathcal{H}_\A$ and $\mathcal{H}_\B$, respectively. Here are $s_i$ are nonnegative real numbers. The Schmidt rank of a quantum state is given by the number of non-zero $s_i$-s in the Schmidt decomposition. 

Then the von Neumann bipartite entanglement entropy of $\ket{\psi}$, with respect to the bipartition $(\A, \B)$ is given by
\begin{equation}
    S(\ket{\psi}_{\A, \B})= \sum_{i=1}^{\mathsf{min}\left(2^{|\A|}, ~2^{|\B|}\right)} -s_i \log s_i.
\end{equation}
There are other measures of entanglement, like the logarithm of the Schmidt rank of a quantum state, with respect to a bipartition. Schmidt rank width is a natural generalization of the Schmidt rank of a quantum state. 

\begin{definition}[Schmidt rank width]
Consider an $n$-qubit state $\ket{\psi}$. Consider trees $T$ with exactly $n$ leaves where the maximum degree of each vertex is $3$. Any edge $e$ in the tree splits the leaves of the tree into two sets $\A_e$ and $\B_e$, depending on the two connected components of $T - e$, where $T - e$ is a graph with edge $e$ removed. Then, the Schmidt rank width of $\ket{\psi}$ is given by
\begin{equation}
    \srw 
    (\ket{\psi}) = \underset{T}{\mathsf{min}}~\underset{e \in T}{\mathsf{max}}~ \log_2~ \mathsf{r}_{\A_e, \B_e},
\end{equation}
where $r_{\A_e, \B_e}$ is the Schmidt rank with respect to the bipartition $\mathsf{r}_{\A_e, \B_e}$. 
\end{definition}

Qualitatively, it measures how large a linear combination we need, in the worst case, to write down the quantum state $\ket{\psi}$ for the best ``tree-like'' decomposition of the state. The more ``succinct'' and more ``tree-like'' the state is, the less is the resource overhead in simulating the state by tree tensor networks \cite{Markov2008, Shi2006, Nest2007a}. 
Therefore, Schmidt-rank width is a measure of the worst-case simulation cost, provided we can figure out an optimal tree decomposition. 

\subsubsection{Entanglement width}
In this paper, we measure entanglement entropy in terms of \emph{entanglement width}, defined in \cite{Nest2006}. 
\begin{definition}[Entanglement width]
Consider an $n$-qubit state $\ket{\psi}$. Consider trees $T$ with exactly $n$ leaves where the maximum degree of each vertex is $3$. Any edge $e$ in the tree splits $n$ qubits into two sets, $\A_e$ and $\B_e$, depending on the two connected components of $T - e$, where $T - e$ is a graph with edge $e$ removed. Then the entanglement width of $\ket{\psi}$ is given by
\begin{equation}
\label{ewidth_def}
    \ew(\ket{\psi}) = \underset{T}{\mathsf{min}}~\underset{e \in T}{\mathsf{max}}~ S\left(\ket{\psi}_{\A_e, \B_e}\right).
\end{equation}
\end{definition}
The qualitative interpretation of this measure is the same as that of the Schmidt rank width: it measures how entangled the state is across ``tree-like'' bipartitions. 

\subsection{Relation to classical simulations}
Entanglement width is important because the runtime of tensor network simulations of a graph state depends on its entanglement width. To formalize this, let us state the following lemma.
\begin{lemma}[\cite{Nest2007a}]
\label{entanglement width and simulations}
Let $\ket{G}$ be an $n$ qubit graph state. Then, using tree-tensor networks, $p_x (G, \mathsf{U})$ can be exactly computed in $\text{poly}\left(n, 2^{\ew(\ket{G})}\right)$ time, for any choice of $\mathsf{U}$.

Additionally, using tree-tensor networks, the output distribution $\mathcal{D}(G, \mathsf{U})$ can be sampled from in  $\text{poly}\left(n, 2^{\ew(\ket{G})}\right)$ time, for any choice of $\mathsf{U}$.
\end{lemma}

An immediate consequence of Lemma \ref{entanglement width and simulations} is that, if the entanglement width of a graph state is upper bounded by $\mathcal{O}(\log n)$, classical simulations are efficient. Also note that implicit in Lemma \ref{entanglement width and simulations} is the fact that we can efficiently find the optimal tree decomposition for $\ket{G}$ in polynomial time, when the entanglement width is logarithmically bounded.

\subsection{Width-measures in graph theory}
\label{appendix: width measures}
Here, we define the notions of tree width, clique width, and rank width. We had merely referenced these measures in the main text: we look at them in detail below.
\subsubsection{Tree width}
First, we look at a measure called ``tree width.'' The definitions are taken from \cite{Diestel2017}. Intuitively, tree width measures how ``similar'' a graph is to a tree. The tree width of any tree is $1$. 
\begin{definition}[Tree decomposition]
A \emph{tree decomposition} of a graph $G = (V, E)$ is a tree $T$ such that the following properties hold.
\begin{enumerate}[(a)]
    \item Each vertex $i$ of $T$ is labelled by a subset $B_i \subset V$. Each such vertex is called a ``bag''.
    \item The two vertices corresponding to every edge in $E$ are both in at least one $B_i$.
    \item For every vertex $u \in V$, the subtree of $T$ consisting of all the ``bags'' containing $u$ is connected.
\end{enumerate}
\end{definition}
\begin{definition}
The \emph{width} of a tree decomposition $T$ is given by $\underset{i}{\mathsf{max}} |B_i| -1$.
\end{definition}

\begin{definition}[Tree width]
The \emph{tree width} of a graph $G$, denoted by $\tw(G)$, is the minimum number $t$ such that there exists a tree decomposition with width at most $t$.
\end{definition}

Many tasks that are $\NP$-hard in general are efficiently solvable in graphs with bounded tree width \cite{Chekuri2014}. Graph states with logarithmically bounded tree width can be efficiently simulated under arbitrary local rotations \cite{Markov2008} if we know the corresponding tree decomposition.

Although the tree width is a very useful concept in these regards, it can be unbounded in ``dense'' graphs, i.e.\ graphs that have lots of edges. However, these graphs can otherwise have a lot of symmetry and many problems can still be easy to solve in these graphs. For example, the complete graph on $n$ vertices has a tree width of $n-1$, but a lot of problems are easy when restricted to just the complete graph. Furthermore, deciding whether the tree width of a graph is at most $m$ for a given integer $m$ is $\NP$-complete \cite{Arnborg1987}.

With that motivation in mind, some generalizations of tree width become necessary.

\subsubsection{Clique width}
 The definition is taken from \cite{Courcelle2000}.
\begin{definition}[Clique width]
The \emph{clique width} of a graph $G$, denoted by $\cw(G)$, is defined as the minimum number of labels needed to construct $G$ with the following operations:
\begin{enumerate}[(a)]
\item Creating a new vertex $v$ that is labeled by an integer $i$.
\item If there are two graphs $G$ and $H$ which are already constructed, taking a disjoint union of these two graphs to create a new graph.
\item Creating a new edge between all vertices of label $i$ and all vertices of label $j$.
\item Changing the label of a vertex.
\end{enumerate}
\end{definition}
Note that the clique width is a generalization of tree width, in the sense that graphs with bounded tree width also have bounded clique width. However, the clique width can remain bounded even when the tree width blows up. For example, the clique width of the complete graph is $2$. In that sense, it is a ``more useful'' metric than the tree width. However, deciding whether the clique width of a graph is at most $m$ for a given integer $m$ is $\NP$-complete \cite{Fellows2009}. So, practically, using this measure might be difficult.

We now define a generalization of clique width that does not have this problem.

\subsubsection{Rank width}
 The definitions are taken from \cite{Oum2017}.
\begin{definition}[Cut rank]
Let $M$ be the $|V| \times |V|$ adjacency matrix of the graph $G$. The \emph{cut rank} of $A \subseteq V$ is the rank of the submatrix of $M$ with row labels corresponding to $A$ and column labels corresponding to $V \setminus A$.
\end{definition}
\begin{definition}[Rank decomposition]
A \emph{rank decomposition} of $G = (V,E)$ is a pair $(T, L)$ where $T$ is a binary tree and $L$ is a bijection from $V$ to the leaves of the tree.
\end{definition}
For a particular rank decomposition $(T, L)$ of a graph $G$, any edge $e$ in the tree $T$ splits $V$ into two parts, $\A_e$ and $\B_e$, corresponding to the two connected components of $T - e$, where $T - e$ is the tree $T$ with edge $e$ removed. The \emph{width} of an edge $e$ is the cut-rank of $\A_e$ (which is equivalent to the cut-rank of $\B_e$). 
\begin{definition}
The \emph{width of the rank decomposition} $(T, L)$ is the maximum width of an edge in $T$. 
\end{definition}
\begin{definition}[Rank width]
The \emph{rank width} of a graph $G$, denoted by $\rw(G)$, is the minimum width of a rank decomposition of $G$.
\end{definition}

Rank width generalizes clique width, in the sense that graphs with bounded clique width also have bounded rank width. Additionally, there is a polynomial time algorithm to decide whether the rank width of a graph is at most $k$, for a given integer $k$ \cite{Courcelle2000}.

\subsubsection{Inter-relations between width measures}
Here, we state a few inter-relations between width measures. Proofs can be found in \cite{Courcelle2000, Gurski2000, Oum2017}.

First, we upper bound clique width in terms of tree width.
\begin{lemma}[\cite{Corneil2001}]
\label{clique width}
For a graph $G$,
\begin{equation}
    \cw(G) \leq 3\cdot2^{\tw(G) - 1}.
\end{equation}
\end{lemma}
Next, we lower bound clique width in terms of the tree width.
\begin{lemma}[\cite{Gurski2000}]
\label{tree width and clique width}
Let $G$ be an $n$--vertex graph such that it does not have the complete bipartite graph $K_{t \times t}$ as a subgraph, for some value of $t$. Then,
\begin{equation}
    \tw(G) \leq 3\cdot \cw(G)\cdot(t-1) - 1.
\end{equation}
\end{lemma}
Now, we relate rank width and clique width.
\begin{lemma}[\cite{Oum2017}]
\label{rank width}
For a graph $G$,
\begin{equation}
  \rw(G) \leq \cw(G) \leq 2^{\rw(G)+1 }-1.
  \end{equation}
\end{lemma}
Now, we state a relation between the clique width of a graph and its complement.
\begin{lemma}[\cite{Courcelle2000}]
\label{complement}
For a graph $G$,
\begin{equation}
    \frac{1}{2}  \cdot \cw(G) \leq \cw(\overline{G}) \leq 2 \cdot \cw(G),
\end{equation}
where $\overline{G}$ is the complement of $G$.
\end{lemma}
Finally, we state one more lemma about how operations like vertex deletion do not increase the clique width. This helps us in independently analyzing the clique width of some of our hard graphs.
\begin{lemma}[\cite{Courcelle2000}]
\label{vertex deletion}
Let $G$ be a graph and $H$ be obtained from $G$ by a sequence of vertex deletions. Then,
\begin{equation}
    \cw(H) \leq \cw(G). 
\end{equation}
\end{lemma}
\subsection{Relation to entanglement width}
Here, we relate the width measures we saw from graph theory to the entanglement measures we saw in the study of quantum states. We relate these two using graph states. 

In general, for any state $\ket{\psi}$, since
\begin{equation}
    S(\ket{\psi}_\mathsf{A, B}) \leq \log_2 \left(\mathsf{r}_{\A, \B}\right)
\end{equation}
for any bipartition $(\A, \B)$, we have
\begin{equation}
   \ew(\ket{\psi}) \leq \srw
   (\ket{\psi}).
\end{equation}

However, the situation greatly simplifies for a graph state. We now state a lemma from \cite{Nest2007a} to show this.

\begin{lemma}
\label{graph states entanglement width}
For any graph state $\ket{G}$,
\begin{equation}
   \ew(\ket{G}) = \srw
   (\ket{G}) = \rw(G).
\end{equation}
\end{lemma}
From Lemma \ref{graph states entanglement width}, entanglement width has exactly the same physical interpretation as the Schmidt rank width for a graph state.

Additionally, by exploiting the connections between rank width and entanglement width, we can argue that entanglement width satisfies the following upper and lower bounds, which can be found in \cite{Oum2017}.
\begin{lemma}
For a graph $G$ and the corresponding graph state $\ket{G}$,
\begin{equation}
   \ew(\ket{G}) \leq \cw(G) \leq 2^{\ew(\ket{G})+1} - 1.
\end{equation}
\end{lemma}

\section{Results}
In this section, first we state our easiness and hardness results. Then, we give a formal proof for each of them, one after another. We also give a formal statement and proof of the duality theorem. A description of the proofs was already provided in the main text. Here, we provide more rigorous technical details, for the sake of completeness. 

Finally, we use our hardness results to sharpen the bound on the entanglement width of certain graphs. 

Let $\mathcal{G}_k$ be the set of all $k$--regular graphs. Consider the following task. \\

\begin{task}{RegularGraph$\big[n, k, x, \mathsf{U}\big]$}
Input & A description of an $n$--vertex, $k$--regular graph $G \in \mathcal{G}_k$, $x \in \{0, 1\}^n$, and a description of the last layer of local rotations $\mathsf{U}$.\\
Output & An inverse polynomial multiplicative error estimate of $p_x(G, \mathsf{U})$.\\
\end{task}

\begin{theorem}[Easy cases]
\label{easy case}
For $k \in \{1, 2, n-3, n-2, n-1 \}$, $\ew(\ket{G})$ is a constant for any $G \in \mathcal{G}_k$ and \textsc{RegularGraph$\big[n, k, x, \mathsf{U}\big]$} is solvable in classical polynomial time.
\end{theorem}

\begin{theorem}[Hard cases]
\label{hard case}
For $3 \leq k \leq n-4$, \textsc{RegularGraph$\big[n, k, x, \mathsf{U}\big]$} is $\# \P$-hard, for any $x \in \{0, 1\}^n$, assuming the $\PH$ does not collapse to $\BPP^{\NP}$.
\end{theorem}

We delegate the proof of these two theorems to \cref{section: easy cases} and \cref{section: hard cases}, respectively. Two corollaries are evident.

\begin{corollary}
\label{second corollary}
For $3 \leq k \leq n-4$, there exists an explicit, efficient construction of a family of $n$--vertex, $k$--regular graphs $\mathcal{F}$ such that
\begin{equation}
\label{ew1}
    \ew(\ket{G}) = \omega(\log n),
\end{equation}
for any $G \in \mathcal{F}$, assuming that the $\PH$ does not collapse to $\BPP^{\NP}$.
\end{corollary}
\begin{proof}
Follows from Lemma \ref{entanglement width and simulations} and Theorem \ref{hard case}. 
\end{proof}
Note that, in our hardness results, we prove something more fine-grained than what is required for \cref{hard case}. We construct $k$--regular graphs for which it is $\# \P$-hard to compute the output probabilities to inverse polynomial multiplicative precision, under local rotations, and then invoke Lemma \ref{lemma: sampling} to prove Theorem \ref{hard case}.

Consequently, we talk about how we can have weaker complexity theoretic conjectures imply the same bound as that of equation \eqref{ew1} and how we can have sharper bounds by tightening the conjecture, in \cref{conditional}.

We also prove a duality between regimes of low regularity and regimes of high regularity, which serves as a convenient tool for proving our hardness results. The duality, as is later discussed, follows from \cref{duality theorem for prob}.

\begin{theorem}[Duality theorem]
\label{duality theorem for prob}
The complement of an $n$--vertex hexagonal lattice, or an $n \times n$ grid graph $G$ is a resource state for measurement based quantum computing.
\end{theorem}

We prove this for the grid graph, and observe that a similar result holds for the hexagonal lattice. We convert the complement of an $n \times n$ grid graph to an $(n-1) \times (n-1)$ grid graph, by a sequence of vertex deletions and local complementation. To prove hardness, we reduce some of our hard graphs to the complement of a grid graph.

\subsection{Proof of \cref{easy case} (The easy regime)}
\label{section: easy cases}
Let $\mathcal{G}_k$ be the set of all $k$--regular graphs. Note that $\cw(G) = 2$ when $G$ is the complete graph \cite{Kaminski2009}. Additionally, note that
\begin{equation}
    \tw(G) = 1
\end{equation}
for every $G \in \mathcal{G}_1$. This is because an $n$ vertex $1$--regular graph is just $n/2$ disconnected lines, each of unit length---so, each can be thought of as a tree with two nodes and one edge and would have a treewidth of $1$---and treewidth does not change by a disjoint union of graphs with same treewidth.

Also note that, for $k=2$,
\begin{equation}
    \tw(G) = 2
\end{equation}
for every $G \in \mathcal{G}_2$. This is because every $2$--regular graph is a disjoint union of a number of series-parallel graphs\footnote{\label{footnote1} For a technical definition of series-parallel graphs, see \cite{Eppstein1992}.}, each of which has a tree width of $2$. Series-parallel graphs draw their inspiration from series-parallel electrical circuits and consist of connected vertices either in series, or in parallel, or a combination of both.

Hence, by  Lemma \ref{rank width} and Lemma \ref{clique width}, the clique width and rank width of every graph in $\mathcal{G}_1$ and $\mathcal{G}_2$ is bounded by a constant. 

By Lemma \ref{complement}, the complement of a bounded clique width graph also has bounded clique width. Since every $(n-2)$-- and $(n-3)$--regular graph is a complement of some $1$--regular or $2$--regular graph, the clique width and the rank width of every $(n-1)$--regular and $(n-2)$--regular graph is also bounded by a constant. By Lemma \ref{graph states entanglement width}, the entanglement width of any graph state $\ket{G}$ is exactly the rank width of the corresponding graph $G$. Hence, the entanglement width of every $1$--, $2$--, $(n-1)$--, $(n-2)$--, and $(n-3)$--regular graph is bounded.

\subsection{The complete graph revisited}
\label{appendix: complete graph}

We had discussed the recursive method for computing the output probabilities of the complete graph in the main text. In this section, we elucidate the technical details behind the recursive method.

We have just established the that the entanglement width of the complete graph is bounded, and hence, by Lemma \ref{lemma: sampling}, efficient sampling and probability estimation is possible under any local rotation. Here, we give a different proof of the fact that output probabilities of the complete graph can be computed in polynomial time, under any local rotation. Our proof leverages inherent symmetry properties of the complete graph. 

The easiness of computing output probabilities of the complete graph has been established before, by observing that a complete graph has bounded clique width \cite{Nest2006}, by observing that a complete graph is reducible to a star graph under local Clifford rotations and the fact that the star graph has bounded tree width \cite{Markham2007}, or by observing that the complete graph can be reduced by local Clifford rotations \cite{Hein2006} to the $\mathsf{GHZ}$ state which can be efficiently classically simulated. But, to the best of our knowledge, our approach has not been taken before and may be utilized elsewhere where there is Hamming weight symmetry. 

The motivation behind this new approach is that our result rely on symmetry properties of the complete graph: this proof can serve as a refresher of those properties. The one specific property we repeatedly use is the fact that the output probabilities of a complete graph has Hamming weight symmetry---it can be written as a linear combination of polynomially many terms, one for each Hamming weight, such that each of them is efficiently computable. In this easiness proof, we can actually go slightly beyond the setting of Figure \ref{fig:setup} by allowing the use of $e^{i \theta Z_i Z_j}$ gates and not just controlled-$\Z$ gates. In more technical terms, define the following state $\ket{G}$ on a complete graph:
\begin{equation}
    \ket{G} = \frac{1}{\sqrt{2^{n}}} \sum_{z \in \{0, 1\}^{n}} \beta^{\left(\underset{i, j \in [n], i < j}{\sum} z_i z_j \right) } \ket{z},
\end{equation}
where $\beta = e^{-i\theta}$ is a unimodular complex number that depends on $\theta$. Let us use the notation $[n]$ to denote the set $\{1, 2, \ldots, n\}$ and let $|z|$ be the Hamming weight of $z$. 

We could equivalently write $\ket{G}$ as
\begin{equation}
\label{eq1}
    \ket{G} = \sum_{y \in 0}^n c_y \sum_{z \in H_y} \ket{\tilde{z}},
\end{equation}
where 
\begin{equation}
\label{eq3}
\ket{\tilde{z}} = \frac{1}{\sqrt{2^n}} \ket{z},
\end{equation}
$H_y$ is the set of all strings with Hamming weight $y$ and
\begin{equation}
\label{eq2}
    c_y = \beta^{\left(\underset{|z| = y; i, j \in [n]; i < j}{\sum} z_i z_j \right) } = \beta^{y(y-1)/2},
\end{equation}
which, by symmetry, is the same for every $z \in H_y$. This is what we mean by Hamming weight symmetry. 

Now, let us formally state the theorem and prove it.

\begin{theorem}
\label{complete graph}
For $k=n-1$, 
    $p_x(G, \mathsf{U})$ can be computed in polynomial time for any choice of $\mathsf{U}$ and any $x \in \{0, 1\}^n$.
\end{theorem}
\begin{proof}
We first reduce computing $p_x(\mathsf{U})$, for a string $x \in \{0, 1\}^n$ to computing $p_{0^{n}}(\mathsf{V})$, for a particular choice of $\mathsf{V}$. To see this reduction, note that for a string $x \in \{0, 1\}^n$, we can write $p_x(U_l)$ as
\begin{equation}
    p_x(\mathsf{U}) =|\langle 0^n| \mathsf{V} |G\rangle|^2,
\end{equation}
where
\begin{equation}
    \mathsf{V} = U_1\mathsf{X}^{x_1} \otimes U_{2}\mathsf{X}^{x_{2}}\otimes \cdots \otimes U_n \mathsf{X}^{x_n}.
\end{equation}
Then, observing that
\begin{equation}
    p_x(\mathsf{U}) = p_{0^n}(\mathsf{V}),
\end{equation}
we compute $p_x(\mathsf{U})$ using our ability to compute $p_{0^n}(\mathsf{V})$, for any $\mathsf{V}$.

Now, consider computing the following quantity:
\begin{equation}
    \mathsf{Z}[n] = \underset{i=1}{\overset{n}{\otimes}} \left(\bra{0} \text{cos}\left(\frac{\theta_i}{2}\right)e^{-i \phi_i} + \bra{1}\text{sin}\left(\frac{\theta_i}{2}\right)\right)\ket{G}.
\end{equation}
Observe that
\begin{equation}
    p_{0^n}(\mathsf{U}) = \big|\mathsf{Z}[n]\big|^{2}.
\end{equation}
So, if we can compute $\mathsf{Z}[n]$, we can compute $p_{0^n}(\mathsf{U})$. In the next steps, we will show how to compute $\mathsf{Z}[n]$. Let
\begin{equation}
    \ket{\psi_{y, n}} = c_y \sum_{z \in H_y} \ket{\tilde{z}}.
\end{equation}
Note that
\begin{equation}
    \ket{G} = \sum_{y \in 0}^n \ket{\psi_{y, n}}.
\end{equation}
Additionally, define
\begin{widetext}
\begin{equation}
    \mathsf{Z}[n, y] 
    = \underset{i=1}{\overset{n-1}{\otimes}} \left(\bra{0} \text{cos}\left(\frac{\theta_i}{2}\right)e^{-i \phi_i} + \bra{1}\text{sin}\left(\frac{\theta_i}{2}\right)\right) \left(\bra{0} \text{cos}\left(\frac{\theta_n}{2}\right)e^{-i \phi_n} + \bra{1}\text{sin}\left(\frac{\theta_n}{2}\right) \right) \ket{\psi_{y, n}}.
\end{equation}
Note that
\begin{equation}
\begin{aligned}
    &\mathsf{Z}[n, y]  \\ &= \text{cos}\left(\frac{\theta_n}{2}\right)e^{-i \phi_n} \underset{i=1}{\overset{n-1}{\otimes}} \left(\bra{0} \text{cos}\left(\frac{\theta_i}{2}\right)e^{-i \phi_i} + \bra{1}\text{sin}\left(\frac{\theta_i}{2}\right)\right)\ket{\psi_{y, {n-1}}} \\&+ \text{sin}\left(\frac{\theta_n}{2}\right) \frac{c_y}{c_{y-1}} \underset{i=1}{\overset{n-1}{\otimes}} \left(\bra{0} \text{cos}\left(\frac{\theta_i}{2}\right)e^{-i \phi_i} + \bra{1}\text{sin}\left(\frac{\theta_i}{2}\right)\right)\ket{\psi_{y-1, {n-1}}}
    \\&= \text{cos}\left(\frac{\theta_n}{2}\right)e^{-i \phi_n} \mathsf{Z}[n-1, y] + \text{sin}\left(\frac{\theta_n}{2}\right) \frac{c_y}{c_{y-1}} \mathsf{Z}[n-1, y-1].
    \end{aligned}
\end{equation}
\end{widetext}
\setcounter{figure}{0}
\renewcommand{\thefigure}{S\arabic{figure}}
\renewcommand{\theHfigure}{S\arabic{figure}}
\begin{figure*}
    \centering
    \includegraphics[width=16cm]{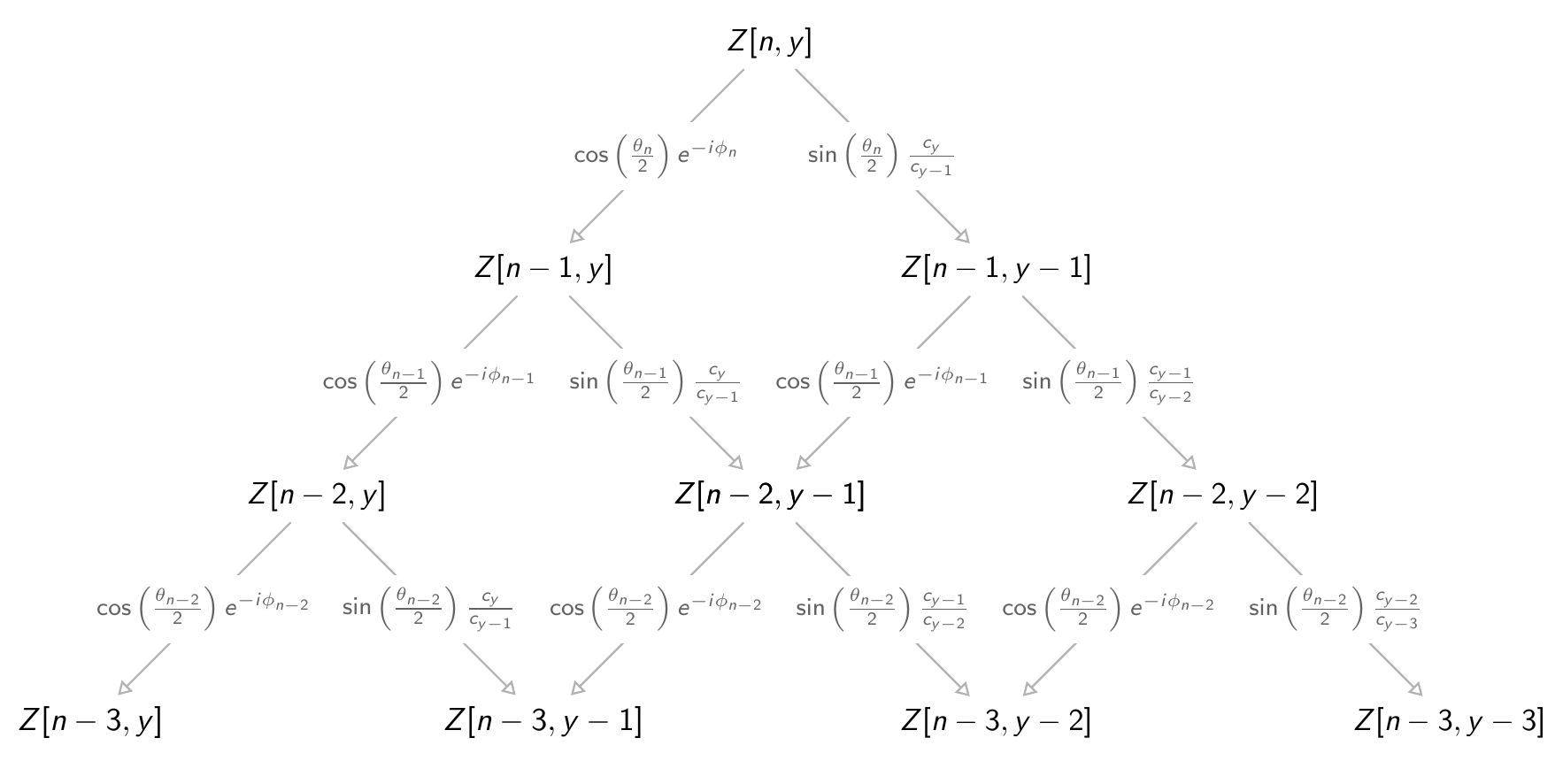}
    \caption{An illustration of how the recursion tree is constructed in the proof of \cref{complete graph}, up to the first three levels.
    \label{appendix: completegraph100}
    }
\end{figure*}

Now, design a recursion tree, a part of which is shown in \cref{appendix: completegraph100}. We start with $\mathsf{Z}[n, y]$ and by multiplying each path with the corresponding path-weights, slowly recurse down the tree. The number of nodes never blows up because there is a lot of overlap during the recursion: for example, in \cref{appendix: completegraph100}, both $\mathsf{Z}[n-1, y]$ and $\mathsf{Z}[n-1, y-1]$ have the same daughter node $\mathsf{Z}[n-2, y-1]$, albeit with different path-weights.

The value of $\mathsf{Z}[n,y]$ is the total value we get in the following way. 
\begin{itemize}
    \item First, we multiply the path weights for each path.
    \item Then, we add up the paths.
\end{itemize}
The number of paths is
\begin{equation}
    2 + 4 + 6 + 8 + \cdots n= \mathcal{O}(n^2).
\end{equation}

\noindent We set
\begin{equation}
\begin{aligned}
    &\mathsf{Z}[n, y] = 1, ~~n < y, \\
    &\mathsf{Z}[0, 0] = 1, \\
    &\mathsf{Z}[0, -1] = 1, \\
    &\mathsf{Z}[1, 0] = \text{cos}\left(\frac{\theta_1}{2}\right)e^{-i\phi_1}.
    \end{aligned}
\end{equation}
to deal with the base cases. Now, do this for every $y$ and add the values to get $\mathsf{Z}[n]$. Since the tree has polynomially many paths that are polynomially large, it can be traversed in polynomial time by standard algorithms.
\end{proof}

\subsection{Proof of \cref{duality theorem for prob} (Duality theorem)}
\label{proofduality}

Let $G$ be an $n \times n$ grid graph. Consider the vertices $a$, $b$, and $c$ as shown in Figure $3$ of the main text ($a$ is the red vertex, and $b$ and $c$ are the green vertices.)

Let $\overline{G}$ be the complement of the grid graph. Follow the following sequence of steps.
\begin{itemize}
    \item Apply a local complementation on $a$ (red vertex).
    \item Vertex delete $a$, $b$, and $c$. 
\end{itemize}

What we get back is a resource state for MBQC because we can get an $(n-1)\times (n-1)$ grid graph from there, just by vertex deletion.

\subsection{Proof of \cref{hard case} (The hard regime)}
\label{section: hard cases}
For every $3 \leq k \leq n/2$, we construct a $k$--regular parent graph such that we can reach the hexagonal graph or the grid graph from that parent graph just by vertex deletion. If we can do that, then by Lemma \ref{hardness grid graph} and Lemma \ref{resource states}, computing the probabilities of the parent graph would be $\# \P$-hard, under inverse polynomial multiplicative precision. Then, the existence of a classical sampler indicates the collapse of the $\PH$ to $\BPP^{\NP}$, by Lemma \ref{lemma: sampling}.

\subsubsection{Technical constructions}
One idea would be to start with the grid-graph or the hexagonal lattice, and just ``reverse engineer'' a construction of the $k$--regular graph, by adding appropriate gadgets to every vertex. This indeed works for small values of $k$, as we see in a demonstration below. 

Consider a family of $n$--vertex, $k$--regular graphs $\mathcal{F}$ and consider the following task.
\begin{task}{$\mathcal{F}$-RegularGraph$\big[n, k, x, \mathsf{U}\big]$}
Input & A description of an $n$--vertex, $k$--regular graph $G \in \mathcal{F}$, $x \in \{0, 1\}^n$, and a description of the last layer of local rotations $\mathsf{U}$.\\
Output & An inverse polynomial multiplicative error estimate of $p_x(G, \mathsf{U})$.\\
\end{task}

\begin{proposition}
\label{hexagonal lattice}
There is an explicit $\mathcal{F}$ such that \textsc{$\mathcal{F}$-RegularGraph$\big[n, 3, x, \mathsf{U} \big]$} is $\# \P$-hard.
\end{proposition}
\begin{proof}
Consider an $n$ vertex hexagonal lattice. Then, we add periodic boundary conditions, which is equivalent to putting this on (or, in more colloquial terms, ``wrapping it around") a torus. This makes it a $3$--regular graph.

Note that one could recover a hexagonal lattice of $\Omega(\sqrt{n}) \times \Omega(\sqrt{n})$ vertices just by a sequence of vertex deletions to ``cut open'' the torus. The proof then follows from the observation that an $n$ vertex hexagonal lattice is a resource state for MBQC.
\end{proof}

\begin{proposition}
\label{gridgraph1}
There is an explicit $\mathcal{F}$ such that \textsc{$\mathcal{F}$-RegularGraph$\big[n, 4, x, \mathsf{U} \big]$} is $\# \P$-hard.
\end{proposition}
\begin{proof}
The proof is the same as that of Proposition \ref{hexagonal lattice}, except we start with an $n$--vertex grid graph, instead of the hexagonal lattice. The proof is also illustrated in Figure $2$(a) of the main text.
\end{proof}

\begin{proposition}
\label{Secondhardgraph}
There is an explicit $\mathcal{F}$ such that \textsc{$\mathcal{F}$-RegularGraph$\big[n, k, x, \mathsf{U} \big]$} is $\# \P$-hard, for $k=n-5$ and $k=n-4$.
\end{proposition}
\begin{proof}
$\mathcal{F}$ is the complement of the graph families constructed in \cref{hexagonal lattice} (for $k=n-4$) and \cref{gridgraph1} (for $k=n-5$). In other words, the hard graphs are the complements of the hexagonal lattice or the grid graph, under closed boundary conditions. Using vertex deletion, one could "cut open" the boundary to reach the complement of a hexagonal lattice or a grid graph respectively, of side length $\Omega(\sqrt{n}) \times \Omega(\sqrt{n})$, which are resource states in MBQC, as proven in \cref{duality theorem for prob}. 
\end{proof}

For other hardness results, we need more involved constructions, with gadgets. Particularly, as explained in the main text, we need to invoke the Gale-Ryser theorem. The technical details of that theorem are provided below.

\begin{proposition}
\label{grid graph}
There is an explicit $\mathcal{F}$ such that \textsc{$\mathcal{F}$-RegularGraph$\big[2m^2, k, x, \mathsf{U} \big]$} is $\# \P$-hard, for any $4 < k \leq m^2$.
\end{proposition}

\begin{proof}
 Start from two $m \times m$ grid graphs. Then, consider closed boundary conditions which is equivalent to putting them on a torus. Thereafter, have edges between the two tori to make every vertex $k$ regular. This is illustrated in Figure $2$(b) of the main text. 
 
 To argue that a valid connection pattern exists, we make use of the Gale-Ryser theorem, the details of which are given below. This suffices for our proof, because one could always delete one torus to reach an object that is a resource state for measurement based quantum computing, by \cref{gridgraph1}. 
 
We need to show that there is a way to connect the vertices such that the final graph is $k$--regular. To boost the regularity of each vertex of the two tori, on either side, to $k$, there are
$m^2$ vertices that we need to add $k-4$ extra edges to.

Since all the edges are ``across'' the two different grid graphs, and no edge is added ``within'' any grid, we can think of our situation as trying to construct a bipartite graph, with $m^2$ vertices on either side, such that the degrees on either side, given by $\mathsf{A}$ and $\mathsf{B}$, are
\begin{equation}
\label{Set}
    \mathsf{A} = \mathsf{B} = \bigg(\underbrace{k-4, \ldots, k-4}_{m^2 ~\text{times}}\bigg).
\end{equation}
Now, we check the conditions of the Gale-Ryser theorem to prove that such degree sequences indeed correspond to a valid bipartite graph. Note that the proof of the Gale-Ryser theorem is constructive. 

Although we do not explicitly construct the graph here, we note that it can be easily done by following the steps of \cite{Ryser2009}.
\begin{center}
\textbf{First condition}
\end{center}
\noindent The first condition of the Gale-Ryser theorem is the following:

\begin{equation}
    \sum_{i=1}^{m^2} a_i = \sum_{i=1}^{m^2} b_i,
\end{equation}
 where $a_i$ is the $i^{\text{th}}$ element of sequence $\mathsf{A}$ and $b_i$ is the $i^{\text{th}}$ element of sequence $\mathsf{B}$.
 
 Note that for our case, $a_i = b_i$ for every $i$. Hence, this statement is trivially true, as can be seen from \eqref{Set}. Now, let us state the second condition.
\begin{center}

\textbf{Second condition}
\end{center}

\noindent For any $1 \leq p \leq m^2$, we need
\begin{equation}
    \sum_{i=1}^p a_i \leq \sum_{i=1}^{m^2} \text{min}(b_i, p).
\end{equation}
\noindent We break the analysis of this condition into three cases, depending on the value of $p$.

\begin{center}
\textbf{First case}
\end{center}

\noindent Let us take
\begin{equation}
    1 \leq p \leq k - 4.
\end{equation}
Then, the RHS becomes $p m^2$. The LHS is upper bounded by $p (k -4)$. Hence, the RHS is larger for any $k \leq m^2$ and the second condition is satisfied.

\begin{center}
\textbf{Second case}
\end{center}

\noindent Let us take 
\begin{equation}
k - 3 \leq p.
\end{equation}
Then, the RHS is
\begin{equation}
    (k-4)m^2.
\end{equation}
The LHS is
\begin{equation}
    (k-4)p.
\end{equation}
Hence, the RHS is always greater than or equal to the LHS.
To conclude, we have identified a way to connect two tori, each with $m^2$--vertices, together such that the resultant graph is $k$-regular, for any $4 \leq k \leq m^2$.
\end{proof}

\begin{remark}
Let $n = 2m^2$. Then, by \cref{gridgraph1}, \textsc{$\mathcal{F}$-RegularGraph$\big[n, k, x, \mathsf{U} \big]$} is $\# \P$-hard, for any $4 < k \leq n/2$. 
\end{remark}

\begin{corollary}
\textsc{RegularGraph$\big[n, k, x, \mathsf{U} \big]$} is $\# \P$-hard, for any $4 < k \leq n/2$. 
\end{corollary}

Now, we can use the duality theorem to argue the hardness of $n$--vertex, $k$--regular graphs with $n/2 + 1 \leq k \leq n-4$. Every hard graph for these families is a complement of the hard graphs we just explicitly constructed.

\begin{proposition}
\label{other hard graph}
There is an explicit, efficiently constructible $\mathcal{F}$ such that \textsc{$\mathcal{F}$-RegularGraph$\big[2m^2, k, x, \mathsf{U} \big]$} is $\# \P$-hard, for any $m^2 + 1 \leq k < 2m^2 - 5$..
\end{proposition}

\begin{proof}
From Proposition \ref{grid graph}, we have already seen how to construct an $\mathcal{H}$ such that \textsc{$\mathcal{H}$-RegularGraph$\big[2m^2, t, x, \mathsf{U} \big]$} is $\# \P$-hard for $4 < t \leq m^2$. For a particular choice of $m$, let $H \in \mathcal{H}$ be a $2m^{2}$--vertex, $t$--regular graph, of the type that was constructed in the proof of \cref{grid graph}. So, $H$ comprises of two copies of a torus connected to each other in a certain way.

Let $G$ be the complement of $H$. $G$ is $k$--regular, where $k = 2m^2 - t - 1$. Then, by deleting all vertices on one copy of the torus, we can reach the complement of a grid graph from $H$, under closed boundary conditions. Then, we can "cut open" the boundary by vertex deletion to reach the complement of a grid graph whose side lengths are $\Omega(\sqrt{n}) \times \Omega(\sqrt{n})$. By \cref{duality theorem for prob}, this is a resource state for measurement based quantum computation, and hence the proof follows. 
\end{proof}

\begin{remark}
Let $n = 2m^2$. Then, by \cref{other hard graph}, \textsc{$\mathcal{F}$-RegularGraph$\big[n, k, x, \mathsf{U} \big]$} is $\# \P$-hard, for any $n/2 \leq k < n-5$. 
\end{remark}

\begin{corollary}
\label{secondcor}
\textsc{RegularGraph$\big[n, k, x, \mathsf{U} \big]$} is $\# \P$-hard, for any $n/2 < k < n-5$. 
\end{corollary}

\begin{proposition}
\textsc{RegularGraph$\big[n, k, x, \mathsf{U} \big]$} is $\# \P$-hard, for any $n/2 < k \leq n-5$.
\end{proposition}

Finally, the proof of \cref{hard case} follows from Proposition \ref{hexagonal lattice}, Proposition \ref{grid graph}, \cref{Secondhardgraph}, \cref{resource states}, and \cref{secondcor}.

\subsection{Bounds on the entanglement width}
We prove some conditional and unconditional lower bounds on the clique width of the graphs constructed in the previous section.

\subsubsection{Unconditional bounds}
\begin{corollary}
\label{cor1}
Let ${G}$ be any $n$--vertex $k$--regular graph constructed in the proofs of any one of Proposition \ref{hexagonal lattice}, \cref{gridgraph1}, Proposition \ref{grid graph}, or Proposition \ref{other hard graph}. Then,
\begin{equation}
    \cw(G) = \Omega(\sqrt{n}).
\end{equation}
Consequently,
\begin{equation}
   \ew(\ket{G}) = \Omega(\log n).
\end{equation}
\end{corollary}
\begin{proof}
For each $G$, we can reach either a grid graph or a hexagonal lattice by vertex deletion, or reach a complement of either. From \cite{Dabrowski2019,Golumbic1999,Kaminski2009}, the $n$--vertex grid graph and the $n$--vertex hexagonal lattice have a clique width $\Omega(\sqrt{n})$. Since, by Lemma \ref{complement},
\begin{equation}
    \cw(\overline{G}) \geq \frac{1}{2} \cw({G}),
\end{equation}
the complement of the grid and the hexagonal lattice also have clique width $\Omega(\sqrt{n})$.

By Lemma \ref{vertex deletion}, since vertex deletion does not increase the clique width, the explicit graphs constructed in Proposition \ref{grid graph} and Proposition \ref{other hard graph} also have clique width $\Omega(\sqrt{n})$. Consequently, the bound on entanglement width follows from Lemma \ref{rank width} and Lemma \ref{graph states entanglement width}.
\end{proof}

\subsubsection{Conditional bounds}
\label{conditional}
\begin{corollary}
\label{entanglement width}
Let $\mathcal{F}$ be any $k$--regular graph family such that \textsc{$\mathcal{F}$-RegularGraph$\big[n, k, x, \mathsf{U}\big]$}, is $\# \class{P}$-hard, for any $x \in \{0, 1\}^n$.
Then, assuming $\BPP \neq \P^{\# \P}$, 
    \begin{equation}
   \ew(\ket{G}) = \omega(\log n),
\end{equation}
for some $G \in \mathcal{F}$.
\end{corollary}

\begin{proof}
Assume the contrary. Then there is an algorithm to solve \textsc{$\mathcal{F}$-RegularGraph$\big[n, k, x, \mathsf{U}\big]$} in classical polynomial time using \cite{Nest2007a}. But, it was assumed that this task is $\# \P$-hard. This implies $\BPP = \P^{\# \P}$, which is a contradiction.
\end{proof}

By using a stronger conjecture, we can get a better lower bound.

\begin{corollary}
\label{entanglement width3}
Let $\mathcal{F}$ be any $k$--regular graph family such that \textsc{$\mathcal{F}$-RegularGraph$\big[n, k, x, \mathsf{U}\big]$}, is $\# \class{P}$-hard, for any $x \in \{0, 1\}^n$.
Then, assuming $\mathsf{ETH}$ \footnote{\label{footnote2} The Exponential Time Hypothesis (ETH) states that there is no solution to $\#\mathsf{3SAT}$ in time $\mathcal{O}(2^{n^{\epsilon}})$, for some constant $\epsilon < 1$ \cite{Beame2016}. Note that $\#\mathsf{3SAT}$ is a $\#\P$-complete function.}, 
    \begin{equation}
   \ew(\ket{G}) = \Omega\left( n^{\delta}\right),
\end{equation}
for some $G \in \mathcal{F}$, and for some $\delta < 1$.
\end{corollary}

\begin{corollary}
\label{entanglement width2}
Let $\mathcal{F}$ be any $k$--regular graph family such that \textsc{$\mathcal{F}$-RegularGraph$\big[n, k, x, \mathsf{U}\big]$}, is $\# \class{P}$-hard, for any $x \in \{0, 1\}^n$ and every $G \in \mathcal{F}$ is reducible to a $\Omega(\sqrt{n}) \times \Omega(\sqrt{n})$ grid graph by a sequence of vertex deletions and local complementations. Then, assuming $\mathsf{SETH} ~$\footnote{ \label{footnote3} The Strong Exponential Time Hypothesis (SETH) states that for every constant $\epsilon >0$, there is a $k$ such that there is no solution to $\#\mathsf{kSAT}$ in time $\mathcal{O}\left(2^{n(1-\epsilon)}\right)$ \cite{Beame2016}. Note that $\#\mathsf{kSAT}$ is a $\#\P$-complete function for every $k \geq 3$.}, 
    \begin{equation}
   \ew(\ket{G}) = \Omega\left( n^{1/2}\right),
\end{equation}
for some $G \in \mathcal{F}$.
\end{corollary}
The proofs are similar to that of Corollary \ref{entanglement width}.
The $n^{1/2}$ comes from the fact that a $\Omega(\sqrt{n}) \times \Omega(\sqrt{n})$ grid graph ``encodes'' a \#\P-hard probability on $\Omega(\sqrt{n})$-many ``logical'' qubits and the fact that entanglement width does not increase under vertex deletions and local complementation. 

\section{Connections to previous work}
In this section, we discuss relations between our setup and what was considered in other relevant works.

\subsection{A discussion of other known classical simulation methods}

Apart from entanglement based methods discussed in the main text, there are known techniques to simulate quantum circuits based on the amount of other resources present in them, like $\mathsf{T}$ gates \cite{Bravyi2016}, stabilizer rank \cite{Bravyi2019b}, or Wigner negativity \cite{Mari2012, Pashayan2015}. We had discussed in the main text why these methods are not adequate to classically simulate our setup: here, we go into more technical details.

In \cite{Bravyi2016}, the runtime of the classical simulator, for both probability estimation and sampling, scales as $\mathcal{O}\left(2^{ct}\right)$, where $t$ is the number of $\mathsf{T}$ gates and $c$ is a constant. For every fixed graph state $\ket{G}$, from Figure $1$(a) of the main text, the presence of a last layer of unitaries means that the number of $\mathsf{T}$ gates, in the worst case, is $n$. Hence, classical simulation is inefficient.

From \cite{Bravyi2019b}, any quantum circuit having Clifford gates and $t$ $\mathsf{T}$ gates can be implemented by a quantum circuit starting with $t$ magic states, each given by
\begin{equation}
    \ket{T} = \frac{1}{\sqrt{2}} \left(\ket{0} + e^{i \pi/4} \ket{1} \right).
\end{equation}
Both the stabilizer and the approximate stabilizer rank of this equivalent circuit are $\mathcal{O}(2^{ct})$, for some constant $c$. Due to local rotations, this factor could be exponentially growing in $n$, in the worst case, for any graph state $\ket{G}$ in Figure $1$(a) of the main text, rendering simulation inefficient. So far, it is not known how to prove better bounds on the stabilizer rank of a circuit or better utilize it as a resource in classical simulations.

Techniques based on Wigner negativity run into similar problems. Note that, when the circuit is only composed of Clifford gates, Wigner negativity, as defined in \cite{Mari2012}, is zero. 
Hence, techniques from Ref.~\cite{Mari2012} can be used to classically sample from the circuit. 
However, local rotations can significantly increase the negativity of the circuit, and the best known techniques for classically sampling from the circuit depend on a metric called ``forward negativity'' \cite{Pashayan2015} being low. 
However, the authors in \cite{Pashayan2015} demonstrated that the forward negativity can, in general, increase exponentially with the number of magic states, which rules out Wigner-negativity-based samplers for our setting. 

For qubit circuits, an additional barrier is that qubit Wigner functions are much harder to define than their qutrit counterparts \cite{Delfosse2015,Raussendorf2017,Raussendorf2020}. 

In conclusion, local rotations render a number of known classical simulation techniques useless for the simulation of graph states. 
Since entanglement is not changed by local rotations, only entanglement-based simulation techniques seem to survive. 
This helps in our analysis and helps us to isolate entanglement as a potential cause of hardness.

\subsection{A discussion of other entanglement measures}
\label{entropy measures}
There are many measures to calculate both bipartite and multipartite entanglement, other than entanglement width, but we show below that it is harder to see what relation they have with hardness, if there is any relation at all. We discuss two popular measures here: the bipartitie von Neumann entropy and the geometric measure. 

The bipartite von Neumann entropy, for a graph state $\ket{G}$, with respect to any bipartition $(\A, \B)$ depends on the size, shape, and location of the bipartition and has the following upper and lower bounds \cite{Hein2004}: 
\begin{multline}
\label{upper bounds}
    \log_2 r = S(\ket{G}_{\A, \B}) \leq \mathsf{Pauli Persistency}(\ket{G}) \\
    \leq \mathsf{minvertexcover}(G),
\end{multline}
where $\mathsf{Pauli~ Persistency}$ is the number of local Pauli measurements required to fully disentangle the graph state, $\mathsf{minvertexcover}(G)$ is the minimal vertex cover of the corresponding graph $G$, and $r$ is the Schmidt rank of $\ket{G}$ across the bipartition $(\A, \B)$. 

By a result from graph theory, for $n$--vertex, $k$--regular graphs,
\begin{equation}
   \mathsf{minvertexcover}(G) \leq \frac{n \cdot k}{k + 1}.
\end{equation}
This implies
\begin{equation}
    S\left(\ket{G}_{\A, \B}\right) \leq \frac{n \cdot k}{k + 1}.
\end{equation}
Note that the upper bound grows with $k$, for a fixed $n$. However, for $k= n - 1$, the von Neumann entropy is $1$ across \emph{any} bipartition \cite{Hein2006}, which indicates that the upper bound is not tight at all. To derive this property of the complete graph, we note that the complete graph reduces to a $\mathsf{GHZ}$ state under local Clifford rotations, and a $\mathsf{GHZ}$ state has von Neumann entropy equal to $1$ for any bipartition. 

On the other hand, for a worst case $2$--regular graph, there is a bipartition $(\A, \B)$ which achieves von Neumann entropy $n/2$. For example, this is achieved by $n$ vertices arranged in a circle, where we fill $\A$ and $\B$ with alternate vertices. At the same time, a bipartition that divides the vertices into two semi-circles has constant von Neumann entropy. Therefore, this metric is very sensitive to the chosen bipartition. 
Since both the complete graph and any $2$--regular graph are easy to classically simulate, but the former has low von Neumann entropy and the latter has high von Neumann entropy across the worst-case cut, it is not clear how von Neumann entropy of the worst case cut relates to classical simulation complexity. 

Alternatively, one could consider von Neumann entropy across the best-case cut $(\A,\B)$ subject to the constraint that $|\A| = \lfloor n/2 \rfloor$. However, for every $n$, one can divide the $n$ qubits into two decoupled sets $\A$ and $\B$ with $|\A| = \lfloor n/2 \rfloor$ and define independent $k$-regular graphs on $\A$ and $\B$, for any constant $k$. The proposed best-case von Neumann entropy would then be zero for all such graphs and will thus clearly not track the simulation complexity.  

Therefore, if one wants to use the von Neumann entropy to track simulation complexity, one needs to make the choice of the cut more cleverly. Indeed, entanglement width, according to \cref{ewidth_def}, is precisely a case of such a cleverly chosen cut: entanglement width is the minimum von Neumann entropy across the most tree-like bipartition, and we show how it tracks simulation complexity. It is plausible that one can find other such variants. 

There are other measures, like the geometric measure of entanglement, given by
\begin{equation}
    S_{\text{geom}}\left(\ket{\psi} \right) = - \log_2~ \underset{\alpha \in \mathcal{P}}{\text{sup}} |\langle \alpha| \psi \rangle|^2,
\end{equation}
where $\mathcal{P}$ is the set of all separable states. For a graph $G$,
\begin{multline}
     \mathsf{m}_\p 
     \leq S_{\text{geom}}\left(\ket{G} \right)\leq \mathsf{Pauli~ Persistency}(\ket{G}) \\
     \leq \mathsf{minvertexcover}(G)\leq \frac{n \cdot k}{k + 1},
\end{multline}
where 
\begin{equation}
\label{bipartition}
    \mathsf{m}_\mathsf{p} = \underset{(\mathsf{A}, \mathsf{B})}{\mathsf{max}} ~\mathsf{m}_\mathsf{p} (\mathsf{A}, \mathsf{B}).
\end{equation}
In equation \eqref{bipartition}, $(\A, \B)$ is a bipartition and $\mathsf{m}_\p (\A, \B)$ is the maximum number of Bell pairs that can be created between $(\A, \B)$, by a bipartite LOCC circuit, comprising only of $\mathsf{CZ}$ gates and local Clifford gates \cite{Markham2007}, when starting from $\ket{G}$. In other words, there can be $\mathsf{CZ}$ gates within each partition, but there should be no $\mathsf{CZ}$ gates across the bipartition. For a complete graph, $\mathsf{m}_{\p}$ is $1$. 

Hence, the known upper and lower bounds on the geometric measure are also not tight, and it is not clear what relation, if any at all, these metrics have with hardness.

\section{Comments on average case hardness}
\label{average-case-hardness}
In the summary section of the main text, we had pointed out that our results generalize to the average case. In this section, we discuss that in detail.

There could be two notions of average case for our setup:
\begin{itemize}
    \item For a particular $G$, and random choice of $\mathsf{U}$.
    \item For a random choice of $G$, and random choice of $\mathsf{U}$, for a particular $k$--regular family.
\end{itemize}

For every hard $k$--regular graph identified in Theorem \ref{hard case}, for $3 \leq k \leq n-4$, computing $p_x(G, \mathsf{U})$ is $\# \P$-hard with high probability over the choice of $\mathsf{U}$, up to an additive error of $2^{-\mathcal{O}(m\log m)}$, using the worst-to-average case reductions in Ref.~\cite{Bouland2022}, where $m$ is the number of gates of the circuit. 

The results are not immediately extendable to a random choice of $G$, because the polynomial interpolation methods of \cite{Bouland2018,Movassagh2019,Bouland2022} may take us beyond $k$--regular graph states. The question then becomes whether it is reasonable to expect average case hardness in this regime.

We assert such an expectation is reasonable and does not violate known results from graph theory, because \emph{most} $k$--regular graphs, for a range of values for $k$, have $\Omega(n)$ clique width and $\Omega(\log n)$ entanglement width. This fact follows straightforwardly from techniques in graph theory, like Lemma \ref{tree width and clique width}, and from works like \cite{Kaminski2009}. But we still state it formally and sketch a proof, just for the sake of completeness.

Note that this does not rule out classical samplers, as there could be a matching upper bound of $\mathcal{O}(\log n)$ for the entanglement width which would make efficient sampling possible by \cref{entanglement width and simulations}---but it makes this setting ``almost'' out of reach of known non-trivial, tree width based samplers. 
\begin{lemma}
\label{high clique width}
Let $G$ be an $n$--vertex, $k$--regular graph picked uniformly at random from the set of all possible $n$--vertex, $k$--regular graphs. Then
\begin{equation}
\label{graphfirsteq}
    \lim_{n \rightarrow \infty} \mathsf{Pr}[\cw(G) = \Omega(n)] = 1,
\end{equation}
when $k = o(n)$, or $k = n - o(n)$. Consequently, under the same conditions as \eqref{graphfirsteq},
\begin{equation}
     \lim_{n \rightarrow \infty} \mathsf{Pr}[\mathsf{ew}\left(\ket{G}\right) = \Omega(\log n)] = 1.
\end{equation}
\end{lemma}
\begin{proof}
First, we prove that random $k$--regular graphs, for $k = o(n)$, have clique width $\Omega(n)$ with high probability. By Lemma \ref{complement}, this means random $n - k - 1$--regular graphs also have a clique width of $\Omega(n)$ with high probability, as these are complements of random $k$--regular graphs. The result on entanglement width then follows from Lemma \ref{rank width} and Lemma \ref{graph states entanglement width}, just as we have seen before.

Let $k = o(n)$. Let $G$ be a random $k$--regular graph. It holds that
\begin{equation}
    \tw(G) = \Theta(n),
\end{equation}
with high probability \cite{Chekuri2014}. From Lemma \ref{tree width and clique width}, 
\begin{equation}
      \tw(G) \leq 3\cdot \cw(G)\cdot(t-1) - 1,
\end{equation}
if $G$ does not have the complete bipartite graph $K_{t \times t}$ as a subgraph. The proof then follows from the following result, which can be seen in Ref.~\cite{Kim2007}:
\begin{equation}
\label{lower bound}
    \lim_{n \rightarrow \infty} \mathsf{Pr}[K_{t\times t}~\text{is not a subgraph of}~G] = 1,
\end{equation}
for any constant $t$, for a random $k$--regular graph with $k = o(n)$.
\end{proof}
Note that Lemma \ref{high clique width} breaks down for the very specific case of when $k = \Theta(n)$. This is an artefact of the proof technique, because equation \eqref{lower bound} breaks down for this case, and we did not find other simple ways to bound the clique width. Nonetheless, we conjecture that there should be a better way to bound the clique width for these cases.

\end{document}